\documentclass[sigconf]{acmart}

\AtBeginDocument{%
  }

\copyrightyear{2025}
\acmYear{2025}
\setcopyright{rightsretained}
\acmConference[MOBIHOC '25]{The Twenty-sixth International Symposium on Theory, Algorithmic Foundations, and Protocol Design for Mobile Networks and Mobile Computing}{October 27--30, 2025}{Houston, USA}
\acmBooktitle{The Twenty-sixth International Symposium on Theory, Algorithmic Foundations, and Protocol Design for Mobile Networks and Mobile Computing (MOBIHOC '25), October 27--30, 2025, Houston}

\usepackage[utf8]{inputenc} 
\usepackage[T1]{fontenc}    
\usepackage{hyperref}       
\usepackage{url}            
\usepackage{booktabs}       
\usepackage{amsfonts}       
\usepackage{nicefrac}       
\usepackage{microtype}      
\usepackage{xcolor}
\usepackage{verbatim}
\usepackage{graphicx}
\usepackage{textcomp}
\usepackage{cases}
\usepackage{amsmath,amsfonts}

\usepackage{bm}
\usepackage{fancyhdr,graphicx}
\usepackage{subfigure}
\usepackage[utf8]{inputenc}
\usepackage{soul}
\usepackage[noend]{algorithm, algorithmic}
\usepackage{multirow}
\usepackage{fontenc}
\usepackage{amsthm}
\usepackage{enumitem}
\newtheorem{theorem}{Theorem}
\newtheorem{lemma}{Lemma}

\usepackage{subcaption}
\newtheoremstyle{example_style}  
  {3pt}  
  {3pt}    
  {\normalfont}  
  {}       
  {\bfseries} 
  {.}      
  { }      
  {\thmname{#1}\thmnumber{ #2}\thmnote{ (#3)}}  

\theoremstyle{example_style}
\newtheorem{Example}{Example} 

\begin{document}

\title{TSDM: A Scheduling Policy for Joint Throughput-AoI Optimization in Multichannel Wireless Networks}
\author{Lin Wang}
\email{linwang0420@tamu.edu}
\affiliation{%
  \institution{Department of ECE, Texas A\& M University}
  \city{College Station}
  \state{Texas}
  \country{USA}
}
\author{I-Hong Hou}
\email{ihou@tamu.edu}
\affiliation{%
  \institution{Department of ECE, Texas A\& M University}
  \city{College Station}
  \state{Texas}
  \country{USA}
}

\begin{abstract}

Optimizing for both low Age of Information (AoI) and high throughput is critical for remote sensing applications that rely on multichannel wireless networks. However, jointly optimizing these two metrics is an analytically challenging problem, particularly in systems with heterogeneous and unreliable channels. To address this challenge, we propose TSDM, a Two-Stage Deficit Matching scheduling framework. TSDM is based on a second-order approach that characterizes the performance of each data flow by its mean and temporal variance. In the first stage, TSDM translates the high-level utility maximization objective into a concrete set of target mean and temporal variance statistics for transmissions over each node–channel pair. In the second stage, a low-complexity Weighted Matching Deficit (WMD) rule performs real-time channel assignment. We theoretically prove that TSDM achieves the desired mean and temporal variance for each flow. Furthermore, we conduct extensive simulations on two open joint throughput-AoI optimization problems. In both cases, TSDM significantly outperforms existing scheduling policies.

\end{abstract}

\ccsdesc[500]{Networks~Network performance modeling; Network performance analysis; Packet scheduling}
\keywords{Wireless networks, age of information, throughput, scheduling}
\maketitle

\section{INTRODUCTION} \label{section:intro}

A growing number of critical applications, from drone surveillance to large-scale environmental monitoring, depend on the timely collection of status updates from spatially distributed sensors. These systems typically operate over multiple wireless channels that exhibit heterogeneous and unreliable conditions due to factors like frequency differences and signal propagation effects\cite{goldsmith2005wireless}. To ensure high performance, a scheduler must intelligently allocate these limited channel resources among a larger number of sensor nodes. Moreover, success in these applications hinges on balancing two competing objectives: maintaining data freshness with low Age of Information (AoI) and ensuring sufficient data with high throughput~\cite{guo2021scheduling}.

While scheduling for AoI minimization, including in multi-channel settings, has been extensively studied, the joint optimization of throughput and AoI remains underexplored. To address this challenge, this paper investigates a new network utility maximization (NUM) problem in multi-channel systems where the utility function of each node depends on both its throughput and AoI performance. The goal is to design a low-complexity scheduling policy that dynamically allocates channel resources based on real-time system states, effectively balancing throughput and AoI performance.

Solving this NUM problem of joint throughput-AoI optimization directly is challenging because it requires explicitly addressing three challenges: the competition for resources among multiple nodes, the trade-off between throughput and AoI, and the stochastic nature of packet deliveries. To overcome these challenges, we propose a novel scheduling framework called Two-Stage Deficit Matching (TSDM). TSDM leverages a second-order analytical approach that captures not only the mean but also the temporal variance of the packet delivery process. This allows us to characterize the achievable throughput-AoI performance in multi-channel systems with heterogeneous unreliable channels.

Furthermore, TSDM decomposes the challenging joint optimization problem into two stages. In the first stage, the high-level utility maximization objective is mapped to a concrete set of target mean and temporal variance statistics for each node-channel pair. In the second stage, a low-complexity Weighted Matching Deficit (WMD) scheduling rule dynamically allocates channels to nodes via bipartite matching, steering the system toward the target statistics derived in the first stage.


We provide both theoretical guarantees and simulation studies. Theoretically, we prove that WMD guarantees the desired mean and temporal variance for each node-channel pair, ensuring system stability and target performance achievement. We further establish the optimality of TSDM for the NUM problem in node-homogeneous systems. Simulation results across multiple scenarios, including utility maximization under throughput constraints and weighted proportional fairness, demonstrate that TSDM achieves near-optimal performance while consistently outperforming existing baselines.


The main contributions of this paper are summarized as follows.
\begin{itemize}

\item We study the joint throughput-AoI NUM problem in multi-channel wireless networks with heterogeneous unreliable channels. 

\item We propose a two-stage scheduling framework that leverages second-order analysis to derive the desired operating point and employs a low-complexity matching-based rule to achieve it. We further prove that the desired point is indeed the optimal point in node-homogeneous systems.
\item Extensive simulations across diverse scenarios show that TSDM closely approaches theoretical performance and consistently outperforms state-of-the-art scheduling baselines.

\end{itemize}

The remainder of this paper is organized as follows. Section~\ref{section:relatedwork} reviews recent studies on wireless scheduling.
Section~\ref{section:model} introduces the system model and formulates the joint throughput-AoI optimization as the NUM problem in multi-channel systems.
Section~\ref{section:secorderappro} presents the second-order analytical framework in multi-channel systems.
Section~\ref{section:policy} introduces the proposed TSDM scheduling policy.
Section~\ref{section:achievability} characterizes the achievable throughput-AoI region under the WMD scheduling rule.
Section~\ref{section:specialcase} establishes the optimality of TSDM under node-homogeneous systems.
Section~\ref{section:simulation} presents simulation results and performance evaluation across different settings.
Finally, Section~\ref{section: conclusion} concludes the paper.
\section{RELATED WORK}\label{section:relatedwork}

Extensive studies have investigated AoI minimization in single-source systems~\cite{firstaoi,minaoiharq,sun2017update}, focusing on optimal scheduling policies under system constraints. However, in practical wireless networks, multiple distributed source nodes transmit updates simultaneously.

Hsu et al.~\cite{hsu2018age,hsu2019scheduling} formulate AoI minimization as an MDP or restless bandit problem and design Whittle index policies that achieve near-optimal performance, while Tripathi et al. ~\cite{tripathi2022optimizing}and Ramakanth et al. ~\cite{ramakanth2024monitoring} study correlated sources and exploit correlation to improve freshness. However, these works do not explicitly consider channel unreliability due to fading.

To account for such impairments, some works extend these models to unreliable channels with stochastic transmission errors~\cite{stocasticarrival,modiano,lin,kadota2018scheduling,sharedonoff}, while some works model the channel 
as a binary ON/OFF process~\cite{guo2024aoi,guo2022theory}. In particular, Kadota et al.~\cite{modiano} propose low-complexity scheduling policies to minimize AoI under throughput constraints, and Wang and Hou~\cite{lin} characterize the fundamental trade-off between throughput and AoI in unreliable wireless networks. Wang et al. ~\cite{csma2} and Tripathi et al.~\cite{freshcsma} further design distributed CSMA-based protocols to approximate centralized scheduling decisions. However, these works primarily focus on single-channel systems, which do not reflect practical multi-channel wireless networks.

For multi-channel systems, prior works~\cite{multichannel1,multichannel2} often model each channel as a binary ON/OFF process and assume that instantaneous channel states are perfectly known before scheduling. In practice, however, channel states are typically unknown at decision time, and scheduling must rely only on statistical channel knowledge. Pan et al.~\cite{pan2022age} study a hybrid two-channel system, but their model is restricted to a fixed setting with one reliable and one unreliable channel, and does not generalize to arbitrary multi-channel systems. Chi et al.~\cite{chi2023aoi} and Zou et al.~\cite{partialindex} both 
study AoI minimization in heterogeneous multi-channel systems, proposing Whittle-index-based policies via restless multi-armed bandit 
formulations. However, neither work incorporates throughput objectives, and the iterative index computation in both approaches may incur 
significant overhead in large-scale systems.

\section{SYSTEM MODEL}\label{section:model}


We consider a time-slotted wireless system where a Base Station (BS) schedules status updates from $N$ sensor nodes over $M$ orthogonal uplink channels, with $M < N$, as shown in Fig.~\ref{fig:systemmodel}. The set of nodes is denoted by $\mathcal{I} = \{1, 2, \dots, N\}$, and the set of available channels by $\mathcal{J} = \{1, 2, \dots, M\}$. Time is slotted and indexed by $t \in \{ 1, 2, \dots\}$. We adopt the standard collision channel model\cite{sun2019closed}. In each time slot, the BS can schedule one node to transmit its status update on each channel, i.e., up to $M$ nodes can transmit simultaneously. Moreover, each node can be scheduled at most once per slot.

The schedule decision in each time slot can be formulated as a matching problem on a bipartite graph $G(t)=(\mathcal I,\mathcal J,\mathcal E)$, where the left vertex set $\mathcal I$ represents the nodes and the right vertex set $\mathcal J$ represents the channels. Each edge $(i,j)\in\mathcal E$ corresponds to a node–channel pair.
A feasible scheduling decision corresponds to a matching 
in $G(t)$, i.e., a set of $M$ edges selected from $\mathcal{E}$ such that each node appears in at most one edge and each channel 
appears in exactly one edge.
We denote the set of all feasible matchings by $\mathcal{K} = \{k_1, k_2,\dots, k_K\}$, where $K$ is the total number of feasible matchings.

The channel is potentially unreliable. Each scheduled transmission from node $i \in \mathcal{I}$ on channel $j \in \mathcal{J}$ succeeds independently with probability $p_{ij} \in (0,1]$, capturing heterogeneous channel conditions due to channel fading. Although $p_{ij}$ remains constant over time, it may vary across different node–channel pairs. We define $Z_i(t)$ as an indicator variable, where $Z_i(t)=1$ denotes that node $i$ successfully delivers an update to the BS at time $t$, and $Z_i(t)=0$ otherwise. Similarly, let $Z_{ij}(t)=1$ denote a successful transmission from node $i$ on channel $j$ at time $t$, and $Z_{ij}(t)=0$ otherwise. Clearly, $Z_i(t)=\sum_{j=1}^MZ_{ij}(t)$.
We adopt the generate-at-will model~\cite{kadota2018scheduling}, where each node scheduled by the BS at time $t$ will generate and send an update containing its status at time $t$. 

To evaluate the performance of each node $i$, we adopt two key metrics: throughput and AoI, which quantify the delivery rate and the timeliness of status updates, respectively.
The throughput of node $i$ is defined as the long-term average number of successful deliveries, given by: $\liminf\limits_{T \to \infty} \frac{1}{T} \sum_{t=1}^T Z_i(t)$. In our model, the AoI of node $i$ at the BS is denoted by $a_i(t)$, which is defined as the time elapsed since the last successfully delivered update from node $i$. The AoI of node $i$ evolves according to the following update rule:
\begin{equation}
    \label{eq:aoi_update}
    a_i(t+1) =
    \begin{cases} 
        1, & \text{if } Z_i(t) = 1, \\
        a_i(t) + 1, & \text{otherwise}.
    \end{cases}
\end{equation}
The long-term average AoI of node \( i \) is then defined as\\ 
 $\limsup\limits_{T \to \infty} \frac{\sum_{t=1}^{T} a_i(t)}{T}.$
\begin{figure}[h]
    \vspace{-10pt}
    \begin{center}
        \includegraphics[width=0.75\linewidth]{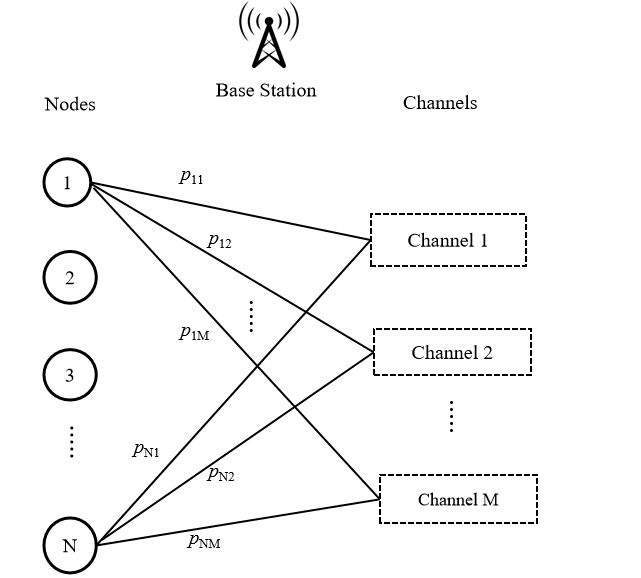}
    \end{center}
    \caption{System model}
    \label{fig:systemmodel}
\end{figure}

The objective of this paper is to develop a scheduling policy that effectively balances throughput and AoI in multi-node wireless networks with multiple heterogeneous and unreliable channels. For each node $i$, let $m_i$ denote its long-term average throughput 
and $h_i$ denote its average AoI. We associate a utility function $U_i(m_i, h_i)$ to node $i$, where $U_i$ is increasing in $m_i$ and decreasing in $h_i$.
Let $\mathcal{R}$ denote the achievable throughput-AoI region, i.e., the set of all performance vectors $\{(m_i, h_i)\}_{i=1}^N$ that can be realized by some scheduling 
policy. The overall objective is to maximize the aggregate utility over the achievable region:
\begin{align}
\max \quad & \sum_{i=1}^N U_i(m_i, h_i) \\
\text{s.t.} \quad & \{(m_i, h_i)\}_{i=1}^N \in \mathcal{R}.
\end{align}
To illustrate the versatility of this framework, we present two representative instances of $U_i$.
\begin{Example}[Utility maximization with soft throughput constraints]
\label{minaoiwithtp}
Kadota et al.~\cite{modiano} first studied the problem of minimizing AoI under hard throughput constraints in a single-channel system. However, when throughput requirements are overly strict or infeasible, enforcing hard constraints may lead to undesirable system behavior. We consider an extension to multi-channel systems where throughput constraints are relaxed and violations are penalized. Let $v_i = (q_i - m_i)^+$ denote the throughput violation of node $i$, 
where $(x)^+ = \max\{x,0\}$, and let $C(\cdot)$ be a nonnegative increasing cost function. The utility function is: $U_i(m_i, h_i) = -(C\big((q_i-m_i)^+\big) + h_i),$
and the optimization problem becomes:
\begin{align}
\max \quad 
& -\sum_{i=1}^N \Big( C\big((q_i - m_i)^+\big) + h_i \Big) \\
\text{s.t.} \quad 
& \{(m_i,h_i)\}_{i=1}^N \in \mathcal{R}.
\end{align}
\end{Example}

\begin{Example}[Weighted proportional fairness in throughput and AoI]
\label{proportional_fairness}
Proportional fairness is a widely adopted objective in network optimization. It is well known that maximizing $\sum_i \log m_i$ achieves proportional fairness in throughput. In the context of status update systems, it is desirable to balance both information freshness and update delivery across nodes. The utility function is: $U_i(m_i, h_i) = \alpha_i \log m_i - \beta_i \log h_i,$
where $\alpha_i, \beta_i > 0$ are weights that reflect the relative 
importance of throughput and AoI for node $i$. The optimization problem becomes:
\begin{align}
\max \quad
& \sum_{i=1}^N \left( \alpha_i\log m_i -\beta_i\log h_i \right) \\
\text{s.t.} \quad
& \{(m_i,h_i)\}_{i=1}^N \in \mathcal{R}.
\end{align}
\end{Example}

\section{SECOND-ORDER APPROXIMATION AND
CHANNEL-WIDE TEMPORAL VARIANCE}\label{section:secorderappro}
In this section, we review recent results that apply second-order techniques to AoI optimization and performance bounds in single-channel systems. We then employ these results to establish several fundamental properties for multichannel systems.

Guo et al.~\cite{guo2022theory} have proposed a framework of second-order analysis to study AoI. In second-order analysis, the packet delivery process $\{Z_i(t)\}$ is characterized by its \emph{mean} $\mu_i$ and \emph{temporal variance} $\sigma^2_i$, which are defined as:
\begin{equation*}
\label{eq:mean}
\mu_i := \lim_{T \to \infty} \frac{\sum_{t=1}^T Z_i(t)}{T}, 
\quad \text{and} \quad
\sigma_i^2 := \mathbb{E}\!\left[\left(\lim_{T \to \infty}
\frac{\sum_{t=1}^T Z_i(t) - T\mu_i}{\sqrt{T}}\right)^2\right].
\end{equation*}

Clearly, $\mu_i$ is also the throughput of node $i$. Moreover, Guo et al.~\cite{guo2022theory} have shown that AoI can be well-approximated as a function of $\mu_i$ and $\sigma_i^2$:
\begin{equation}
    \label{eq:aoi_appro}
  \limsup\limits_{T \to \infty} \frac{\sum_{t=1}^{T} a_i(t)}{T}\approx\frac{1}{2} \left(\frac{\sigma_i^2}{\mu_i^2} + \frac{1}{\mu_i} \right) + \frac{1}{2}.
\end{equation}



For all our theoretical analysis, we assume the approximation in Eq.~\eqref{eq:aoi_appro} is exact, and we will evaluate the accuracy of the approximation in Section~\ref{section:simulation}.

We now discuss how to extend this framework to multi-channel systems. We define the mean and temporal variance of $\{Z_{ij}(t)\}$ as:
\begin{equation*}
\label{eq:mean_var_multi}
\mu_{ij} := \lim_{T \to \infty} \frac{\sum_{t=1}^T Z_{ij}(t)}{T},
\quad
\sigma_{ij}^2 := 
\lim_{T \to \infty}
\mathbb{E}\!\left[
\left(
\frac{\sum_{t=1}^T Z_{ij}(t) - T\mu_{ij}}{\sqrt{T}}
\right)^2
\right],
\end{equation*}
for each node $i$ on channel $j$. 


Since $Z_i(t)=\sum_{j=1}^MZ_{ij}(t)$, we clearly have $\sum_{j=1}^{M} \mu_{ij} = \mu_i$, for all $i$ under any scheduling policy.

However,  the relationship between $\sigma_i^2$ and 
$\{\sigma_{ij}^2\}_{j=1}^M$ is more complex: $\sigma_i^2$ cannot be expressed solely as a function of $\{\sigma_{ij}^2\}_{j=1}^M$. 
This is because the computation of $\sigma_i^2$ requires knowledge of the cross-channel covariance terms 
$\mathrm{Cov}(Z_{ij}(t), Z_{ij'}(t))$ for any two distinct channels $j \neq j'$. 
Handling this coupling across channels is one of the main technical contributions of this work, which we will discuss later.

To proceed, we first characterize the fundamental limitations on $\mu_{ij}$ and $\sigma^2_{ij}$ imposed by the unreliable nature of wireless transmission.

Recall that a transmission of node $i$ on channel $j$ succeeds with
probability $p_{ij}$.
Since each node $i$ can be scheduled on at most one channel, and each
channel $j$ is assigned to exactly one node in each time slot, the long-term throughputs must satisfy:
\begin{equation}\label{eq:neccon}
   0 \leq \sum_{j=1}^{M} \frac{\mu_{ij}}{p_{ij}} \leq 1, \quad \forall i,
\qquad
\sum_{i=1}^{N} \frac{\mu_{ij}}{p_{ij}} = 1, \quad \forall j. 
\end{equation}

Besides, Wang and Hou~\cite{lin} established a bound on the temporal variance in single-channel systems. By applying their result to each channel individually, we obtain the following per-channel variance $\sigma_j^2$ constraints :
\begin{equation}\label{eq:channel_variance_constraint}
    \sum_{i=1}^{N} \sqrt{\frac{\sigma_{ij}^2}{p_{ij}^2}}  \geq \sqrt{\sigma_j^2}=
\sqrt{\sum_{i=1}^{N} \frac{\mu_{ij}}{p_{ij}} \left( \frac{1}{p_{ij}} - 1 \right)},  \quad \forall j.
\end{equation}

\section{A TWO-STAGE SCHEDULING POLICY}\label{section:policy}
In this section, we propose a scheduling policy called \emph{Two-Stage Deficit Matching} (TSDM).
The proposed policy consists of two stages: In the first stage, for the throughput-AoI requirement vector $\{(m_i,h_i)\}_{i=1}^N$, we map them to a corresponding set of per node-channel pair parameters $\{(\mu_{ij},\sigma_{ij}^2)\}$. In the second stage, we employ a simple matching rule to deliver the desirable $\{(\mu_{ij},\sigma_{ij}^2)\}$.

We first discuss the first stage. Given $\{(m_i,h_i)\}_{i=1}^N$, we will find $\{(\mu_{ij},\sigma_{ij}^2)\}$ that satisfy the following conditions:
\begin{align}
&\mu_i \geq \ m_i,\quad\frac{1}{2} \left( \frac{\sigma_i^2}{\mu_i^2} + \frac{1}{\mu_i} \right) + \frac{1}{2} \leq h_i,\forall i,\label{eq:aoi2}\\
&\sum_{j=1}^M \mu_{ij}=\mu_i, \forall i,\label{eq:sumofmean2}\\
&0<\sum_{j=1}^M\frac{\mu_{ij}}{p_{ij}}<1, \forall i,\quad \sum_{i=1}^{N}\frac{\mu_{ij}}{p_{ij}}=1, \forall j,\label{eq:somean2}\\
&\sum_{i=1}^{N} \sqrt{\frac{\sigma_{ij}^2}{p_{ij}^2}} = \sqrt{\sum_{i=1}^{N} \frac{\mu_{ij}}{p_{ij}} \left( \frac{1}{p_{ij}} - 1 \right)},  \forall j ,\label{eq:sovariance2}\\
&\sigma_i^2=\sum_{j=1}^M\sigma^2_{ij}, \forall i,\label{eq:sumofvariance}\\
&\mu_{ij} > 0,\sigma^2_{ij}>0, \forall i,j.\label{eq:geq02}
\end{align}

Condition Eq.~\eqref{eq:aoi2} enforces the throughput and AoI requirements, respectively. 
Conditions Eq.~\eqref{eq:sumofmean2}--Eq.~\eqref{eq:sovariance2} are based on the constraints introduced in the previous section. 
Eq.~\eqref{eq:sovariance2} is taken as the equality form of Eq.~\eqref{eq:channel_variance_constraint}, corresponding to the tight boundary case of the per-channel variance constraint. 
To handle the cross-channel coupling in $\sigma^2_i$, we introduce Eq.~\eqref{eq:sumofvariance}, which decomposes the total variance $\sigma_i^2$ of node $i$ into a sum of per-channel variance components $\sigma_{ij}^2$. This decomposition decouples the cross-channel covariance terms and enables tractable optimization, as we will see in Sec ~\ref{section:specialcase}, this additional condition still preserves optimality under a homogeneous system. Finally, Eq.~\eqref{eq:geq02} is a non-degeneracy condition ensuring that all per-channel throughput and variance targets are strictly positive.







 Combined with the example application problems in Sec~\ref{section:model}, the first stage effectively solves an optimization problem to find the optimal$\{(\mu_{ij},\sigma_{ij}^2)\}$ within the region defined by Eq.~\eqref{eq:aoi2}--Eq.~\eqref{eq:geq02}.

In the second stage, based on the obtained solution from the optimization problem, we apply a \emph{weighted matching deficit} (WMD) rule to perform real-time channel assignment.
This WMD rule drives the system toward the desired operating point by achieving the target throughput and variance $\{(\mu_{ij},\sigma_{ij}^2)\}$ obtained from the first stage in the long run.

To introduce the WMD rule, we first define the \emph{deficit} of a node \( i \) on channel $j$ at time \( t \)  as follows:
\begin{equation}\label{eq:d_i}
d_{ij}(t) := \frac{t \mu_{ij} - \sum_{\tau=1}^{t} Z_{ij}(\tau)}{\sqrt{\sigma_{ij}^2}}.
\end{equation}

We note that $t\mu_{ij}$ represents the number of deliveries required for node $i$ on channel $j$ by time $t$ to achieve a mean throughput of $\mu_{ij}$, while $\sum_{\tau=1}^{t} Z_{ij}(\tau)$ denotes the actual number of deliveries completed up to time $t$.
Hence, $t\mu_{ij}-\sum_{\tau=1}^{t} Z_{ij}(\tau)$ represents how much the actual deliveries for node $i$ deviates from the target mean $\mu_{ij}$ on channel $j$. We further normalize this deviation by the target variance $\sqrt{\sigma_{ij}^2}$.

We then define the channel-wide normalized deficit on channel $j$ as:
\begin{equation}\label{eq:D_j}
D_j(t) :=
\frac{\sum_{i=1}^{N} \sqrt{\frac{\sigma_{ij}^2}{p_{ij}^2}}\, d_{ij}(t)}
{\sum_{i=1}^{N}\sqrt{\frac{\sigma_{ij}^2}{p_{ij}^2}}}.
\end{equation}
 
The corresponding weight for each node--channel pair is given by: $W_{ij}(t) :=  \sqrt{{\frac{\sigma_{ij}^2}{p_{ij}^2}}}(d_{ij}(t) - D_j(t)).$
For any matching $k$, its matching weight is defined as: $Y_k(t):=\sum\limits_{(i,j)\in k}W_{ij}(t).$
The WMD rule selects the max-weight matching at time $t$: $k^*(t) = \arg\max_{k\in\mathcal K} Y_k(t),$ and schedules transmissions accordingly.

\section{ACHIEVABLE REGION UNDER THE WMD RULE}\label{section:achievability}

In this section, we show that the WMD policy can achieve every point in the region described by Eq.~\eqref{eq:aoi2}--Eq.~\eqref{eq:geq02}.

To analyze the performance of the TSDM policy, we first establish a structural property of the node-channel pair weight process $\{W_{ij}(t)\}$. 
Specifically, we show that the weight process is zero-sum across nodes for each channel.

\begin{lemma}\label{lemma:sumofweight}
$ \sum_{i=1}^NW_{ij}(t)=0,$  for all $j$.
\end{lemma}
\begin{proof} 
For each channel $j$:
\begin{align}
    \sum_{i=1}^NW_{ij}(t) &= \sum_{i=1}^N\sqrt{{\frac{\sigma_{ij}^2}{p_{ij}^2}}}(d_{ij}(t)-D_j(t)) \nonumber\\
    &= \sum_{i=1}^N\sqrt{{\frac{\sigma_{ij}^2}{p_{ij}^2}}}d_{ij}(t)-\sum_{i=1}^N\sqrt{{\frac{\sigma_{ij}^2}{p_{ij}^2}}}\frac{\sum_{i=1}^{N} \sqrt{{\frac{\sigma_{ij}^2}{p_{ij}^2}}} d_{ij}(t)}{\sum_{i=1}^{N} \sqrt{{\frac{\sigma_{ij}^2}{p_{ij}^2}}}} = 0 \nonumber\label{weightedw}
\end{align}  
\end{proof}
Next, we establish a relationship between the largest node-channel pair weight and the maximum matching weight.

\begin{lemma}\label{lemma:yandw}
At any time slot $t$, the maximum matching weight satisfies:
\[
Y_{k^*}(t) \geq \frac{N-M}{N-1} \max_{(i,j)\in\mathcal{E}} W_{ij}(t),
\]
where $k^* = \arg\max_{k \in \mathcal{K}} Y_k(t)$ is the max-weight matching.
\end{lemma}
\begin{proof}
Without loss of generality, assume $W_{11}(t) = \max\limits_{(i,j)\in\mathcal{E}} W_{ij}(t)$,
i.e., node-channel pair $(1,1)$ achieves the largest edge weight at time $t$.
We further assume that the max-weight matching among all matchings 
containing $(1,1)$ assigns node $i$ to channel $i$ for all $i=1,\ldots,M$. 
Denote this matching as $\tilde{k}$, so that:
\[
Y_{\tilde{k}}(t) = \sum_{i=1}^{M} W_{ii}(t) = W_{11}(t) + \sum_{i=2}^{M} W_{ii}(t).
\]

Since $Y_{\tilde{k}}(t)$ is the maximum weight over all matchings 
containing $(1,1)$, it is at least as large as the average matching 
weight over all such matchings. For each channel $j\in\{2,\ldots,M\}$, 
averaging uniformly over all $N-1$ remaining nodes yields an average 
per-channel contribution of $\frac{1}{N-1}\sum_{i=2}^{N}W_{ij}(t)$.
By Lemma~\ref{lemma:sumofweight}, $\sum_{i=1}^N W_{ij}(t)=0$, 
so $\sum_{i=2}^N W_{ij}(t) = -W_{1j}(t)$, and the average 
contribution of channel $j$ is $-\frac{W_{1j}(t)}{N-1}$.
Summing over all remaining channels gives:
\[
Y_{\tilde{k}}(t) \geq W_{11}(t) - \frac{1}{N-1}\sum_{j=2}^{M}W_{1j}(t).
\]
Since $W_{11}(t) = \max\limits_{(i,j)\in\mathcal{E}}W_{ij}(t) \geq W_{1j}(t)$ 
for all $j$, we have $\sum_{j=2}^M W_{1j}(t)\leq (M-1)W_{11}(t)$.
Substituting yields:
\[
Y_{\tilde{k}}(t) \geq W_{11}(t) - \frac{M-1}{N-1}W_{11}(t) 
= \frac{N-M}{N-1}W_{11}(t) 
= \frac{N-M}{N-1}\max_{(i,j)\in\mathcal{E}}W_{ij}(t).
\]
Since global maximum matching weight $Y_{k^*}(t)\geq Y_{\tilde{k}}(t)$, this completes the proof.
\end{proof}




We are now ready to analyze the performance of the WMD rule. Consider the Lyapunov function: $L(t):=\frac{1}{2}\sum_{j=1}^M \sum_{i=1}^NW_{ij}^2(t)$ to do the drift analysis. 
Let $\Delta d_{ij}(t):=d_{ij}(t)-d_{ij}(t-1)$, $\Delta D_j(t):= D_j(t)-D_j(t-1)$, and $\Delta L(t):= L(t)-L(t-1)$, we can then derive the expected one-step Lyapunov drift as follows:
\begin{align}
    &E[\Delta L(t)] := E[L(t) - L(t - 1) \mid H^{t-1}]  \nonumber \\
    =&E\Bigg[\frac{1}{2} \sum_{j=1}^M\sum_{i=1}^NW_{ij}^2(t) - \frac{1}{2}  \sum_{j=1}^M\sum_{i=1}^N W_{ij}^2(t-1)\;\Bigg| H^{t-1}\Bigg]  \nonumber \\
    =&E\Bigg[ \sum_{j=1}^M\sum_{i=1}^N \sqrt{{\frac{\sigma_{ij}^2}{p_{ij}^2}}}W_{ij}(t-1) \left( {\Delta d_{ij}(t)} - \Delta D_j(t) \right) \nonumber \\
    &+\frac{1}{2}  \sum_{j=1}^M\sum_{i=1}^N{{\frac{\sigma_{ij}^2}{p_{ij}^2}}}\left( {\Delta d_{ij}(t)} - \Delta D_j(t) \right)^2 \;\Bigg| H^{t-1}\Bigg]  \nonumber \\
    \leq& B + E\Bigg[ \sum_{j=1}^M\sum_{i=1}^N\sqrt{{\frac{\sigma_{ij}^2}{p_{ij}^2}}} W_{ij}(t-1) \Delta d_{ij}(t) \nonumber \\
    &-  \sum_{j=1}^M\sum_{i=1}^N\sqrt{{\frac{\sigma_{ij}^2}{p_{ij}^2}}}W_{ij}(t-1) \Delta D_j(t) \;\Bigg| H^{t-1}\Bigg]  \nonumber \\
    =&B + E\Bigg[ \sum_{j=1}^M \sum_{i=1}^N\sqrt{{\frac{\sigma_{ij}^2}{p_{ij}^2}}}W_{ij}(t-1) \Delta d_{ij}(t) \;\Bigg| H^{t-1}\Bigg],  \label{eq:delta_L}
\end{align}
where $H^{t-1}$ is the system history up to time $t-1$ and \(B\) is a constant. The last two steps hold because \(\Delta d_{ij}(t)\) and \(\Delta D_j(t)\) are bounded, and according to Eq.~\eqref{eq:somean2}, we can derive that \( E[\Delta D_j(t)] = 0 \).






\begin{lemma}\label{lemma:positive_recurrent}
For any vector $\{(\mu_{ij}, \sigma_{ij}^2)\}$ that satisfies Eqs.~\eqref{eq:aoi2}--\eqref{eq:geq02}, the system state process $\{W_{ij}(t)\}$ under the WMD scheduling rule is positive recurrent. 
\end{lemma}
\begin{proof}\label{proof:positive_recurrent}

By the design of WMD, the BS schedules the matching $k^*$ with the highest 
weight at time $t$. Hence, The expected increment of $d_{ij}(t)$ is:
\begin{equation}
E[\Delta d_{ij}(t)] =
\begin{cases}
    \frac{\mu_{ij}-p_{ij}}{\sigma_{ij}}, & \quad \text{if}  (i,j) \in k^*, \\
\frac{\mu_{ij}}{\sigma_{ij}}, & \text{otherwise}.
\end{cases}
\label{eq:delta_di}
\end{equation}

We then have
\begin{align*}
&E \left[ \sum_{j=1}^M\sum_{i=1}^{N}\sqrt{\frac{\sigma_{ij}^2}{p_{ij}^2}}W_{ij}(t) \Delta d_{ij}(t) \right]  =\sum_{j=1}^M\sum_{i=1}^{N} W_{ij}(t) \frac{\mu_{ij}}{p_{ij}}-\sum_{(i,j)\in k^*}W_{ij}(t)
\end{align*}
By the feasibility of the first stage and the Birkhoff--von Neumann 
theorem, there exists a probability distribution $\{\pi_k\}_{k\in\mathcal{K}}$ 
over matchings such that
\[
\frac{\mu_{ij}}{p_{ij}} = \sum_{k=1}^K \pi_k \,\mathbf{1}\{(i,j)\in k\},
\quad \forall\, i,j,
\]
where $\pi_k > 0$ and $\sum_{k=1}^K \pi_k = 1$. 


Then, we have:

$E\left[ \sum_{j=1}^M\sum_{i=1}^{N}\sqrt{\frac{\sigma_{ij}^2}{p_{ij}^2}}W_{ij}(t) \Delta d_{ij}(t) \right]=\sum_{r=1}^K \pi_r Y_{k_r}(t) - Y_{k^*}(t)$.

Applying the telescope with the ordering 
$Y_{k_1}(t)\geq Y_{k_2}(t)\geq\cdots\geq Y_{k_K}(t)$ 
and $k_1 = k^*$:
\begin{align} \nonumber
\sum_{r=1}^K \pi_r Y_{k_r}(t) - Y_{k_1}(t)
&= \sum_{r=1}^{K-1}(Y_{k_r}(t)-Y_{k_{r+1}}(t))
\left(\sum_{j=1}^{r}\pi_j - 1\right). \label{eq:telescope}
\end{align}

Let 
$0<\varepsilon < \frac{1}{2}\min\limits_r\pi_r .$
Then, for any \( r \in \{1, 2, \dots, K-1\} \):
\begin{equation}
\sum_{j=1}^{r} \pi_j - 1 = -\sum_{j=r+1}^{K} \pi_j \leq -\pi_{r+1} < -\varepsilon.
\label{eq:ineq1}
\end{equation}
Therefore,

\begin{align}   \nonumber
&E\left[ \sum_{j=1}^M\sum_{i=1}^{N}\sqrt{\frac{\sigma_{ij}^2}{p_{ij}^2}}W_{ij}(t) \Delta d_{ij}(t) \right]< -\varepsilon \left( Y_{k_1}(t) - Y_{k_K}(t) \right)<-\varepsilon  Y_{k_1}(t) 
\label{eq:expectation_ineq}
\end{align}
The last step holds since $Y_{k_K}(t) \leq 0$, 
which follows from Lemma~\ref{lemma:sumofweight} that $\sum_{k \in \mathcal{K}} Y_k(t)=\sum\limits_{(i,j)\in k}W_{ij}(t)= 0$ .

Thus, if $\max_{k} Y_k(t) > \frac{B}{\varepsilon},$ then \( E[\Delta L] < 0 \).



By Lemma~\ref{lemma:yandw}, we can derive that if $\max_{(i,j)\in\mathcal{E}} W_{ij}(t)>\frac{N-1}{N-M}\frac{B}{\varepsilon},$ then $\max_k Y_k(t) > \frac{B}{\varepsilon},$
which implies $\mathbb{E}[\Delta L] < 0$. Therefore, under our scheduling rule, the state $\{W_{ij}(t)\}$ follows a Markov process with negative drift outside a compact set. By the Foster-Lyapunov theorem, the system-wide Markov process is positive recurrent.
\end{proof}

We now show that the WMD rule achieves the desired throughputs and AoIs that return from the first stage, thereby establishing Theorem ~\ref{theorem:verification}.\\
\begin{theorem}\label{theorem:verification}
Assume that the first stage conditions return a feasible solution $\{(\mu_{ij},\sigma_{ij}^2)\}$. Then, under the WMD rule, the following holds:
\begin{equation}\label{eq:thm_meanvar}
     \mu_i = \lim_{T \to \infty} \frac{\sum_{t=1}^T Z_i(t)}{T}, \quad \sigma_i^2 = \mathbb{E}\!\left[\left(\lim_{T \to \infty}\frac{\sum_{t=1}^T Z_i(t) - T\mu_i}{\sqrt{T}}\right)^2\right].
\end{equation}

\end{theorem}
\begin{proof}\label{proof:verification}
According to Lemma ~\ref{lemma:positive_recurrent}, the system-wide Markov process $\{W_{ij}(t)\}$ is positive recurrent under the WMD policy, then, for each channel $j$, we have:
\begin{equation}\label{eq:T}
\lim_{T \to \infty} \frac{d_{ij}(T)- D_j(T)}{T} \to 0, \quad \forall i,
\end{equation}
\begin{equation}\label{eq:sqrt_T}
\lim_{T \to \infty} \frac{d_{ij}(T) - D_j(T)}{\sqrt{T}} \to 0, \quad \forall i. 
\end{equation}

First, we show the throughput claim in Eq.~\eqref{eq:thm_meanvar}.

Recall that \(d_{ij}(t) = \frac{t \mu_{ij} - \sum_{\tau=1}^{t} Z_{ij}(\tau)}{\sqrt{\sigma_{ij}^2}}.\) and \(D_j(t) = \frac{\sum_{i=1}^N \sqrt{\frac{\sigma_{ij}^2}{p_{ij}^2}} d_{ij}(t)}{\sum_{i=1}^N \sqrt{\frac{\sigma_{ij}^2}{p_{ij}^2}} }\). By Eq.~\eqref{eq:somean2}, we have:

\begin{align}
&\lim_{T \to \infty} \frac{D_j(T)}{T} \nonumber
= \lim_{T \to \infty} 
\frac{\sum_{i=1}^N\sqrt{\frac{1}{p_{ij}^2}}( T {\mu_{ij}} - \sum_{t=1}^T {Z_{ij}(t)})}
{T \sum_{i=1}^N \sqrt{\frac{\sigma_{ij}^2}{p_{ij}^2}}} \\ 
&= \lim_{T \to \infty} \frac{1 - \frac{1}{T} \sum_{i=1}^N \sum_{t=1}^T \frac{Z_{ij}(t)}{p_{ij}}}
{\sum_{i=1}^N \sqrt{\frac{\sigma_{ij}^2}{p_{ij}^2}}}=0\label{eq:mean3}
\end{align}

Since, according to the Law of Large Numbers:

\(\frac{1}{T} \sum_{i=1}^N \sum_{t=1}^T \frac{Z_{ij}(t)}{p_{ij}} \to 1\) as \(T \to \infty\).

Combining Eq.~\eqref{eq:T} and Eq.~\eqref{eq:mean3}, we have:

$\lim\limits_{T \to \infty} \frac{d_{ij}(T)}{T} = \frac{\mu_{ij}}{\sqrt{\sigma_{ij}^2}} - \lim\limits_{T \to \infty} \frac{\sum_{t=1}^T {Z_{ij}(t)}}{{\sqrt{\sigma_{ij}^2}}T} = 0, \quad \text{for all } i.$

Rearranging the terms, we conclude:
\(\lim\limits_{T \to \infty} \frac{\sum_{t=1}^T Z_{ij}(t)}{T} = \mu_{ij}\). Then,

$\sum_{j=1}^M\lim\limits_{T \to \infty} \frac{\sum_{t=1}^T Z_{ij}(t)}{T}=\lim\limits_{T \to \infty} \frac{\sum_{t=1}^T Z_i(t)}{T}=\sum_{j=1}^M\mu_{ij}=\mu_i.$

Next, we show the variance claim in Eq.~\eqref{eq:thm_meanvar}. Following the approach in Wang and Hou~\cite{lin}, we define the projected process $X_j(t)$ for channel $j$ as: $X_j(t) := t - \sum_{\tau=1}^t \sum_{i=1}^N \frac{Z_{ij}(\tau)}{p_{ij}}.$ According to the Martingale Central Limit Theorem,  we can derive:$ \lim\limits_{T \to \infty} \frac{X_j(T)}{\sqrt{T}} \sim \mathcal{N}\left(0, (\sum_{i=1}^N \sqrt\frac{\sigma_{ij}^2}{p_{ij}^2})^2 \right).$

Then:

\begin{align}\label{eq:variance3}
&\lim_{T \to \infty} \frac{D_j(T)}{\sqrt{T}} 
= \lim_{T \to \infty} 
\frac{\sum_{i=1}^N \sqrt{\frac{1}{p_{ij}^2}} \left( T \mu_{ij} - \sum_{t=1}^T Z_{ij}(t) \right)}
{\sqrt{T} \sum_{i=1}^N \sqrt{\frac{\sigma_{ij}^2}{p_{ij}^2}}} \nonumber \\
&= \lim_{T \to \infty} 
\frac{T - \sum_{i=1}^N \sum_{t=1}^T \frac{Z_{ij}(t)}{p_{ij}}}
{\sqrt{T} \sum_{i=1}^N \sqrt{\frac{\sigma_{ij}^2}{p_{ij}^2}}} 
= \lim_{T \to \infty} 
\frac{X_j(T)}{\sqrt{T} \sum_{i=1}^N \sqrt{\frac{\sigma_{ij}^2}{p_{ij}^2}}}.
\end{align}

Therefore, we have: $\lim\limits_{T \to \infty} \frac{D_j(T)}{\sqrt{T}} \sim \mathcal{N}(0,1).$

Furthermore, using Eq.~\eqref{eq:sqrt_T} and Eq.~\eqref{eq:variance3}, we derive:
\begin{align}
&\mathbb{E}\left[\left(\lim_{T \to \infty} \frac{\sum_{t=1}^T Z_{ij}(t) - T \mu_{ij}}{\sqrt{T}}\right)^2\right]= \mathbb{E}\left[\left(\lim_{T \to \infty} \frac{\sigma_{ij} \cdot d_{ij}(T)}{\sqrt{T}}\right)^2\right] \nonumber \\
&= \sigma_{ij}^2 \cdot \mathbb{E}\left[\left(\lim_{T \to \infty} \frac{D_j(T)}{\sqrt{T}}\right)^2\right]= \sigma_{ij}^2.
\end{align}

We can also get:
\begin{align}\label{eq:d_ijto_X}
    \lim\limits_{T \to \infty} \frac{d_{ij}(T)}{\sqrt{T}}\to\lim\limits_{T \to \infty} 
\frac{X_j(T)}{\sqrt{T} \sum_{i=1}^N \sqrt{\frac{\sigma_{ij}^2}{p_{ij}^2}}}.
\end{align}

Define the deficit of each node $i$ as: $d'_i(t):=t\mu_i - \sum_{t=1}^T Z_i(t)=\sum_{j=1}^M(t\mu_{ij}-\sum^t_{\tau=1}Z_{ij}(\tau))$. 

Then, according to Eq.~\eqref{eq:d_i}:
$d'_i(t)=\sum_{j=1}^Md_{ij}(t)\sqrt{\sigma^2_{ij}}$.



Denote $a_{ij}:=\frac{\sqrt{\sigma^2_{ij}}}{ \sum_{i=1}^N \sqrt{\frac{\sigma_{ij}^2}{p_{ij}^2}}}$, then combining Eq.~\eqref{eq:d_ijto_X} and the definition of the temporal variance for node $i$:
\begin{align}
&\mathbb{E}\left[\left(\lim_{T\to\infty}
\frac{\sum_{t=1}^T Z_i(t)-T\mu_i}{\sqrt{T}}\right)^2\right] = \mathbb{E}\left[\left(\lim_{T\to\infty}-\frac{d'_i(T)}{\sqrt{T}}\right)^2\right]\nonumber\\
&= \mathbb{E}\left[\left(\lim_{T\to\infty}-\frac{\sum_{j=1}^Md_{ij}(T)\sqrt{\sigma^2_{ij}}}{\sqrt{T}}\right)^2\right] = \mathbb{E}\left[\left(\lim_{T\to\infty}
\sum_{j=1}^M a_{ij}\frac{X_j(T)}{\sqrt{T}}\right)^2\right]   \nonumber\\
&= \sum_{j=1}^M a_{ij}^2\,\mathbb{E}\left[\left(\lim_{T\to\infty}
\frac{X_j(T)}{\sqrt{T}}\right)^2\right] \nonumber\\
&\quad + 2\sum_{j=1}^{M}\sum_{j'=j+1}^{M} a_{ij}a_{ij'}\,
\mathrm{Cov}\left(\lim_{T\to\infty}\frac{X_j(T)}{\sqrt{T}},\,
\lim_{T\to\infty}\frac{X_{j'}(T)}{\sqrt{T}}\right).
\label{eq:33}
\end{align}

We now show that, under the WMD rule, the limiting projected
processes of distinct channels are uncorrelated.

Write $\Delta X_j(\tau):=X_j(\tau)-X_j(\tau-1)
=1-\sum_{i=1}^N \frac{Z_{ij}(\tau)}{p_{ij}}$, with $X_j(0)=0$.
Given $H^{\tau-1}$, the matching $k^*(\tau)$ is determined and channel
$j$ is assigned to exactly one node $i$, whose transmission succeeds
with probability $p_{ij}$. Hence
\[
\mathbb{E}\!\left[\Delta X_j(\tau)\mid H^{\tau-1}\right]
= 1-\frac{p_{ij}}{p_{ij}}=0,
\]
so $\{X_j(t)\}$ is a martingale with respect to the filtration
$\{H^{t}\}$.

For $j\neq j'$, a matching assigns channels $j$ and $j'$ to two
\emph{distinct} nodes, and transmissions on orthogonal channels
succeed independently. Thus $\Delta X_j(\tau)$ and
$\Delta X_{j'}(\tau)$ are conditionally independent given
$H^{\tau-1}$, which gives
\[
\mathbb{E}\!\left[\Delta X_j(\tau)\,\Delta X_{j'}(\tau)\mid H^{\tau-1}\right]
=\mathbb{E}\!\left[\Delta X_j(\tau)\mid H^{\tau-1}\right]
 \mathbb{E}\!\left[\Delta X_{j'}(\tau)\mid H^{\tau-1}\right]=0 .
\]
For $\tau>s$, $\Delta X_{j'}(s)$ is $H^{\tau-1}$-measurable, so
\[
\mathbb{E}\!\left[\Delta X_j(\tau)\,\Delta X_{j'}(s)\right]
=\mathbb{E}\!\left[\Delta X_{j'}(s)\,
 \mathbb{E}[\Delta X_j(\tau)\mid H^{\tau-1}]\right]=0,
\]
and symmetrically for $\tau<s$. Expanding
$X_j(T)=\sum_{\tau=1}^{T}\Delta X_j(\tau)$, every term vanishes:
\[
\mathbb{E}\!\left[X_j(T)X_{j'}(T)\right]
=\sum_{\tau=1}^{T}\sum_{s=1}^{T}
\mathbb{E}\!\left[\Delta X_j(\tau)\Delta X_{j'}(s)\right]=0 .
\]
Since $\mathbb{E}[X_j(T)]=\mathbb{E}[X_{j'}(T)]=0$, dividing by $T$
yields
\begin{equation}\label{eq:cov}
\mathrm{Cov}\!\left(\lim_{T\to\infty}\frac{X_j(T)}{\sqrt{T}},\;
\lim_{T\to\infty}\frac{X_{j'}(T)}{\sqrt{T}}\right)
=\lim_{T\to\infty}\frac{1}{T}\mathbb{E}\!\left[X_j(T)X_{j'}(T)\right]=0 .
\end{equation}
Accordingly:
\begin{align*}
   & \mathbb{E}\left[\left(\lim_{T \to \infty} 
\frac{\sum_{t=1}^T Z_i(t) - T \mu_i}{\sqrt{T}}\right)^2\right] = \sum_{j=1}^M a_{ij}^2\, \mathbb{E}\left[\left(\lim_{T\to\infty}
\frac{X_j(T)}{\sqrt{T}}\right)^2\right]\\
&=\sum_{j=1}^M\frac{\sigma_{ij}^2}{\left(\sum_{i=1}^N\sqrt{\frac{\sigma_{ij}^2}{p_{ij}^2}}\right)^2}
\cdot \left(\sum_{i=1}^N\sqrt{\frac{\sigma_{ij}^2}{p_{ij}^2}}\right)^2 
= \sum_{j=1}^M\sigma_{ij}^2.
\end{align*}

Then, according to Eq.~\eqref{eq:sumofvariance}:
\begin{align*}
\mathbb{E}\left[\left(\lim_{T \to \infty} \frac{\sum_{t=1}^T Z_i(t) - T \mu_i}{\sqrt{T}}\right)^2\right]&=\sum_{j=1}^M\sigma_{ij}^2=\sigma^2_i
\end{align*}
which completes the proof of Eq.~\eqref{eq:thm_meanvar}.

\end{proof}
\section{OPTIMALITY OF THE TSDM POLICY UNDER NODE-HOMOGENEOUS SYSTEMS}
\label{section:specialcase}

The first stage of our policy finds the optimal $\{(\mu_{ij},\sigma_{ij}^2)\}$ under the constraints Eq.~\eqref{eq:aoi2}--Eq.~\eqref{eq:geq02}. Among these, Eq.~\eqref{eq:aoi2}--Eq.~\eqref{eq:sovariance2} are necessary conditions imposed by the system, whereas Eq.~\eqref{eq:sumofvariance} is an additional constraint introduced by our policy. In this section, we study whether imposing Eq.~\eqref{eq:sumofvariance} negatively impacts system performance. We focus on \emph{node-homogeneous} systems, where all nodes experience the same channel reliability on each channel; that is, for each channel $j$, there exists a $p_j$ such that $p_{ij} = p_j$ for all $i$, while reliabilities may still vary across channels.

Consider an arbitrary scheduling policy $\eta$. We use $\mu_{ij}(\eta)$, $\sigma^2_{ij}(\eta)$, $\mu_i(\eta)$, $\sigma^2_i(\eta)$, $m_i(\eta)$, and $h_i(\eta)$ to denote the values of $\mu_{ij}$, $\sigma^2_{ij}$, $\mu_i$, $\sigma^2_i$, $m_i$, and $h_i$ achieved under $\eta$, respectively. We say $\eta$ is a \emph{non-singular} policy if it satisfies the following two mild conditions:
\begin{enumerate}[label=(\roman*)]
    \item $\mu_{ij}(\eta) > 0$ for all $i\in\mathcal{I}$ and $j\in\mathcal{J}$, i.e., each node is scheduled on each channel with a strictly positive long-run fraction of time; and
    \item $\sum_{j=1}^{M} \frac{\mu_{ij}(\eta)}{p_{ij}} < 1$ for all $i\in\mathcal{I}$, i.e., no node occupies all available scheduling opportunities.
\end{enumerate}

We show that in node-homogeneous systems, the additional constraint Eq.~\eqref{eq:sumofvariance} incurs no performance loss. Specifically, for any utility function $U_i(m_i, h_i)$ that is increasing in $m_i$ and decreasing in $h_i$, TSDM achieves an aggregate utility no smaller than that of any non-singular scheduling policy, in terms of the AoI approximation Eq.~\eqref{eq:aoi_appro}:
\begin{align}\label{optimalutility}
\sum_{i=1}^N U_i(m_i, h_i)
\geq
\sum_{i=1}^N U_i(m_i(\eta), h_i(\eta)),
\quad \text{for all}\, \text{non-singular } \eta.
\end{align}

We begin with a lower bound on the sum of temporal standard deviations $\sum_{i=1}^{N}\sigma_i(\eta)$, where $\sigma_i(\eta) := \sqrt{\sigma_i^2(\eta)}$.

\begin{lemma}\label{lemma:lowerbound}
In node-homogeneous systems where $p_{ij}=p_j$, every non-singular policy $\eta$ satisfies:
\[
\sum_{i=1}^{N}\sigma_i(\eta) \;\ge\;
\sqrt{\sum_{j=1}^{M}(p_j - p_j^2)}.
\]
\end{lemma}
\begin{proof}
Define the aggregate centered delivery process
\[
S(t) = \sum_{i=1}^{N}\bigl(Z_i(t)-\mathbb{E}[Z_i(t)]\bigr).
\]
Since on each channel $j$ exactly one node is scheduled and succeeds with probability $p_j$ regardless of which node is chosen, the per-channel contributions to $S(t)$ are independent across channels. Hence $\mathbb{E}[S(t)]=0$ and
\[
\mathrm{Var}(S(t)) = \sigma_{\mathrm{tot}}^2
:= \sum_{j=1}^{M}(p_j - p_j^2).
\]
Expanding $\mathrm{Var}(S(t))$ over nodes gives:
\begin{align}
\sigma_{\mathrm{tot}}^2
= \sum_{i=1}^{N}\mathrm{Var}(Z_i)
+ 2\sum_{1\le i<i'\le N}\mathrm{Cov}(Z_i,Z_{i'}),
\label{eq:var-sum-ik}
\end{align}
where the $Z_i := Z_i(t)$ are evaluated at a fixed slot $t$.
By the Cauchy--Schwarz inequality, each cross-term satisfies
\[
|\mathrm{Cov}(Z_i,Z_{i'})|\le
\sqrt{\mathrm{Var}(Z_i)\,\mathrm{Var}(Z_{i'})}
= \sigma_i(\eta)\,\sigma_{i'}(\eta).
\]
Substituting into Eq.~\eqref{eq:var-sum-ik}:
\[
\sigma_{\mathrm{tot}}^2
\le \sum_{i=1}^N\sigma_i^2(\eta)
+ 2\sum_{1\le i<i'\le N}\sigma_i(\eta)\sigma_{i'}(\eta)
= \Bigl(\sum_{i=1}^N\sigma_i(\eta)\Bigr)^2.
\]
Taking square roots on both sides completes the proof.
\end{proof}

We now establish the optimality of TSDM.
\begin{theorem}\label{thm:optimality}
In node-homogeneous systems, for any non-singular policy $\eta$, there exists a feasible input $\{(\mu_{ij}, \sigma_{ij}^2)\}$ to TSDM such that the resulting $\{(\mu_i, \sigma_i^2)\}$ satisfies:
\begin{enumerate}[label=(\roman*)]
    \item $\{(\mu_i,\sigma^2_i)\}$ and $\{(\mu_{ij},\sigma^2_{ij})\}$ satisfy Eq.~\eqref{eq:sumofmean2}--Eq.~\eqref{eq:geq02}.
    \item $\mu_i = \mu_i(\eta)$ and $\sigma_i^2 \leq \sigma_i^2(\eta)$, for all $i\in\mathcal{I}$.
    \item Eq.~\eqref{optimalutility} holds.
\end{enumerate}
\end{theorem}

\begin{proof}
Construct the TSDM input parameters as follows:
\[
\mu_i = \mu_i(\eta),\quad \mu_{ij} = \mu_{ij}(\eta), \quad \forall\, i,j,
\]
\[
\sigma_i = \sigma_i(\eta)\cdot\frac{\sigma_{\mathrm{tot}}}{\sum_{k=1}^N\sigma_k(\eta)},
\quad
\sigma_{ij} = \frac{\sigma_i}{\sigma_{\mathrm{tot}}}\sqrt{p_j-p_j^2},
\quad \forall\, i,j.
\]

\textit{Feasibility (part~(i)).}
Since $\eta$ is a feasible policy, the necessary conditions Eq.~\eqref{eq:neccon} ensure that $\mu_i = \mu_i(\eta)$ and $\mu_{ij} = \mu_{ij}(\eta)$ satisfy Eq.~\eqref{eq:sumofmean2} and Eq.~\eqref{eq:somean2}. For Eq.~\eqref{eq:sovariance2}, noting that $\sum_{i=1}^N\sigma_i = \sigma_{\mathrm{tot}}$ by construction:
\[
\sum_{i=1}^N\sigma_{ij}
= \frac{\sqrt{p_j-p_j^2}}{\sigma_{\mathrm{tot}}}\sum_{i=1}^N\sigma_i
= \sqrt{p_j-p_j^2}, \quad \forall\, j.
\]
For Eq.~\eqref{eq:sumofvariance}, using $\sigma_{\mathrm{tot}}^2=\sum_{j=1}^M(p_j-p_j^2)$:
\[
\sum_{j=1}^M\sigma_{ij}^2
= \frac{\sigma_i^2}{\sigma_{\mathrm{tot}}^2}\sum_{j=1}^M(p_j-p_j^2)
= \sigma_i^2, \quad \forall\, i.
\]
Hence $\{(\mu_{ij}, \sigma_{ij}^2)\}$ is a feasible input to TSDM, which then delivers the target $(\mu_i, \sigma_i^2)$ for all $i$ via the WMD rule.

\textit{Variance reduction (part~(ii)).}
By Lemma~\ref{lemma:lowerbound}, $\sum_{k=1}^N\sigma_k(\eta)\ge\sigma_{\mathrm{tot}}$, so
\[
\sigma_i
= \sigma_i(\eta)\cdot\frac{\sigma_{\mathrm{tot}}}{\sum_{k=1}^N\sigma_k(\eta)}
\le \sigma_i(\eta), \quad \forall\, i.
\]
Squaring both sides gives $\sigma_i^2\le\sigma_i^2(\eta)$ for all $i$.

\textit{Utility optimality (part~(iii)).}
Since $\mu_i = \mu_i(\eta)$ and $\sigma_i^2 \le \sigma_i^2(\eta)$, the AoI approximation Eq.~\eqref{eq:aoi_appro} gives $h_i \le h_i(\eta)$ for all $i$. Because $U_i$ is increasing in $m_i$ and decreasing in $h_i$:
\[
U_i(m_i, h_i) \ge U_i(m_i(\eta), h_i(\eta)), \quad \forall\, i.
\]
Summing over all nodes establishes Eq.~\eqref{optimalutility}.
\end{proof}

\section{ SIMULATION RESULTS}\label{section:simulation}
In this section, we present the simulation results for the proposed TSDM framework. We apply our policy to the practical and fundamental problems proposed in Section~\ref{section:model}. In the following simulations, each experiment is conducted over \(1,000,000 \times N\) time slots to ensure a comprehensive evaluation of the scheduling policy. The performance metrics are averaged over 1,000 independent traces to obtain statistically meaningful results. To evaluate the proposed policy, we also include numerical solutions obtained from solving Example~\ref{minaoiwithtp}, ~\ref{proportional_fairness}, which serve as benchmarks and are referred to as the Theoretical results. 
\subsection{Utility Maximization with Soft Throughput Constraints}\label{subsection:softp}

We first study the problem described in Example~\ref{minaoiwithtp}. We consider the problem of minimizing AoI under soft throughput constraints, where each node $i$ has a throughput requirement $q_i$. Any violation incurs a quadratic penalty $C(v_i) = cv_i^2$, where $v_i = (q_i - m_i)^+$ and $c>0$ is a scaling constant chosen to ensure the penalty term is on the same order of magnitude as the AoI term.

As baselines, we consider a multi-channel extension of the Max-Weight scheduling policy~\cite{modiano} and the Partial Index policy~\cite{partialindex}. The Max-Weight scheduling policy keeps track of $ x_i(t+1) = t q_{i} - \sum_{\tau=1}^{t} Z_{i}(\tau)$ for each node $i$, and assigns to every node–channel pair the weight:
 $W_{ij}(t) = \frac{p_{ij}}{2} a_i(t) [a_i(t) + 2] + N^2 p_{ij} x_{i}^+(t)$. It then selects the matching $k$ that maximizes $\sum\limits_{(i,j)\in k}W_{ij}(t)$. Prior work has shown that the Max-Weight policy achieves near-optimal performance in single-channel systems. The Partial Index policy, on the other hand, has been proven asymptotically optimal for AoI minimization in multi-channel systems but does not account for throughput requirements. We therefore include it to quantify the performance loss incurred when throughput constraints are ignored.

We evaluate the three policies in both the homogeneous and heterogeneous settings. In the homogeneous setting, we fix $M=2$ and set $p_{ij}=(0.9,0.3)$ for all $i$. The soft throughput constraint is set as $q_i=\lambda\gamma_i$, where $\gamma_i = \frac{1.6}{N}\sum_{j=1}^{M}\max\limits_{i} p_{ij}$ if $i \leq N/2$ and $\frac{0.4}{N}\sum_{j=1}^{M}\max\limits_{i} p_{ij}$ otherwise, so that $\lambda = 1$ corresponds to the system operating at the maximum achievable capacity. We evaluate three network sizes with $N\in \{10,20,50\}.$ For each network, we evaluate the average Utility, defined as the negative average of the sum of penalty and AoI across different $\lambda$. Simulation results are shown in Fig~\ref{fig:stp}.

\begin{figure*}[t]
    \centering
    \subfigure[$N = 10,M = 2$.]{
       \includegraphics[width=0.3\linewidth]{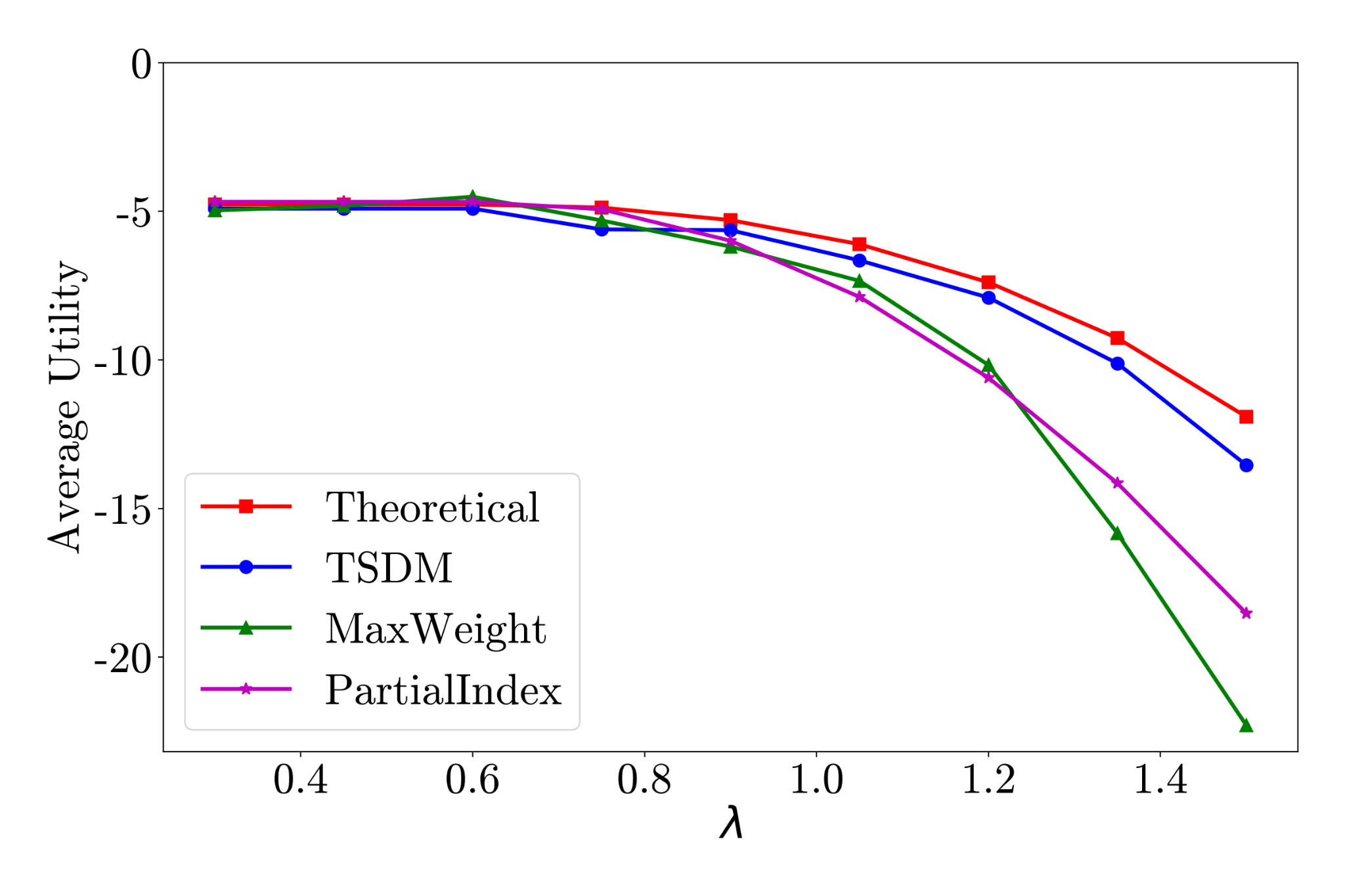}
       \label{fig:hoN10_2}
    }
    \hfill
    \subfigure[ $N = 20,M = 2$.]{
       \includegraphics[width=0.3\linewidth]{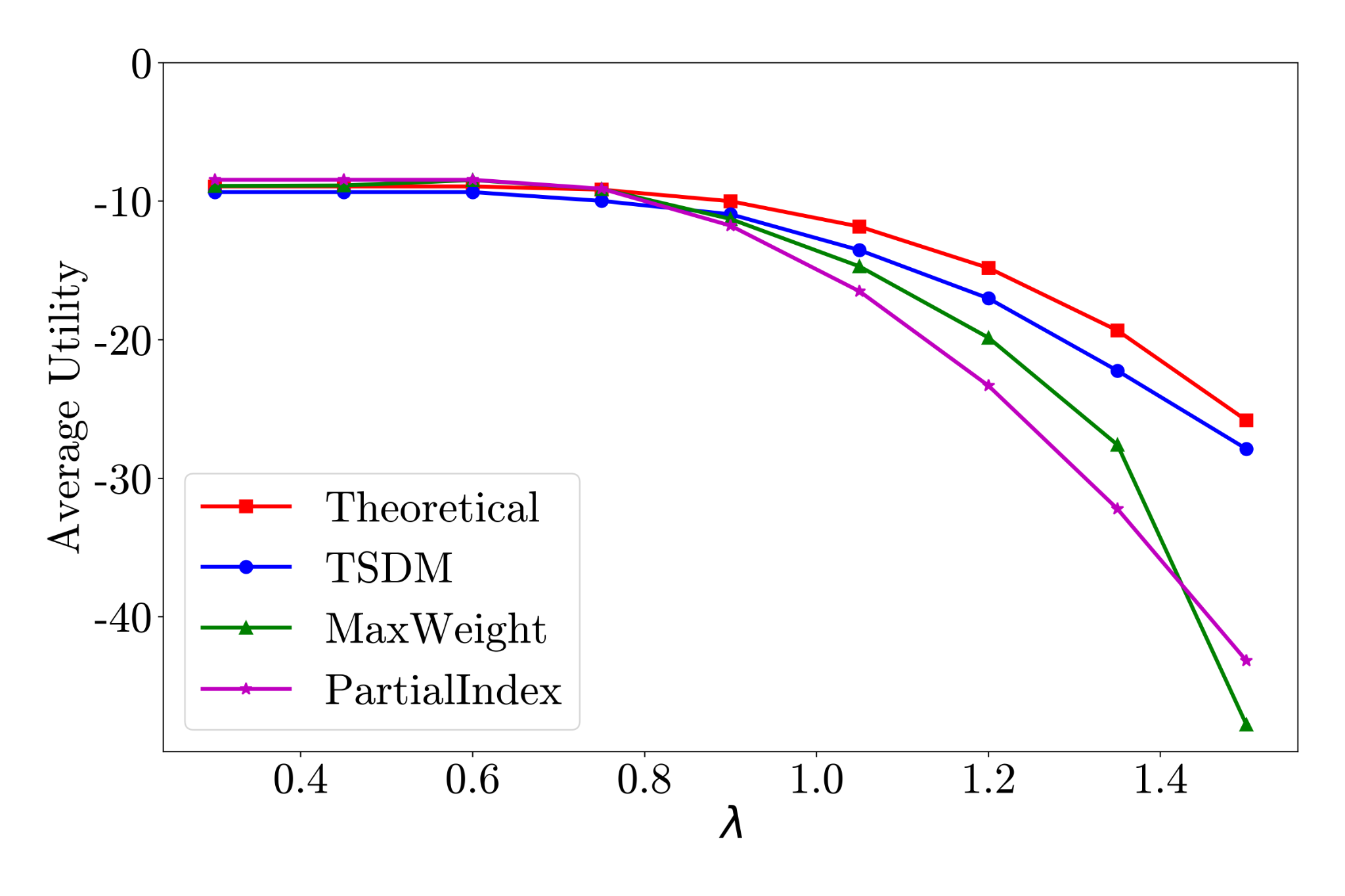}
       \label{fig:hoN20_2}
    }
    \hfill
    \subfigure[$N = 50,M = 2$.]{
       \includegraphics[width=0.3\linewidth]{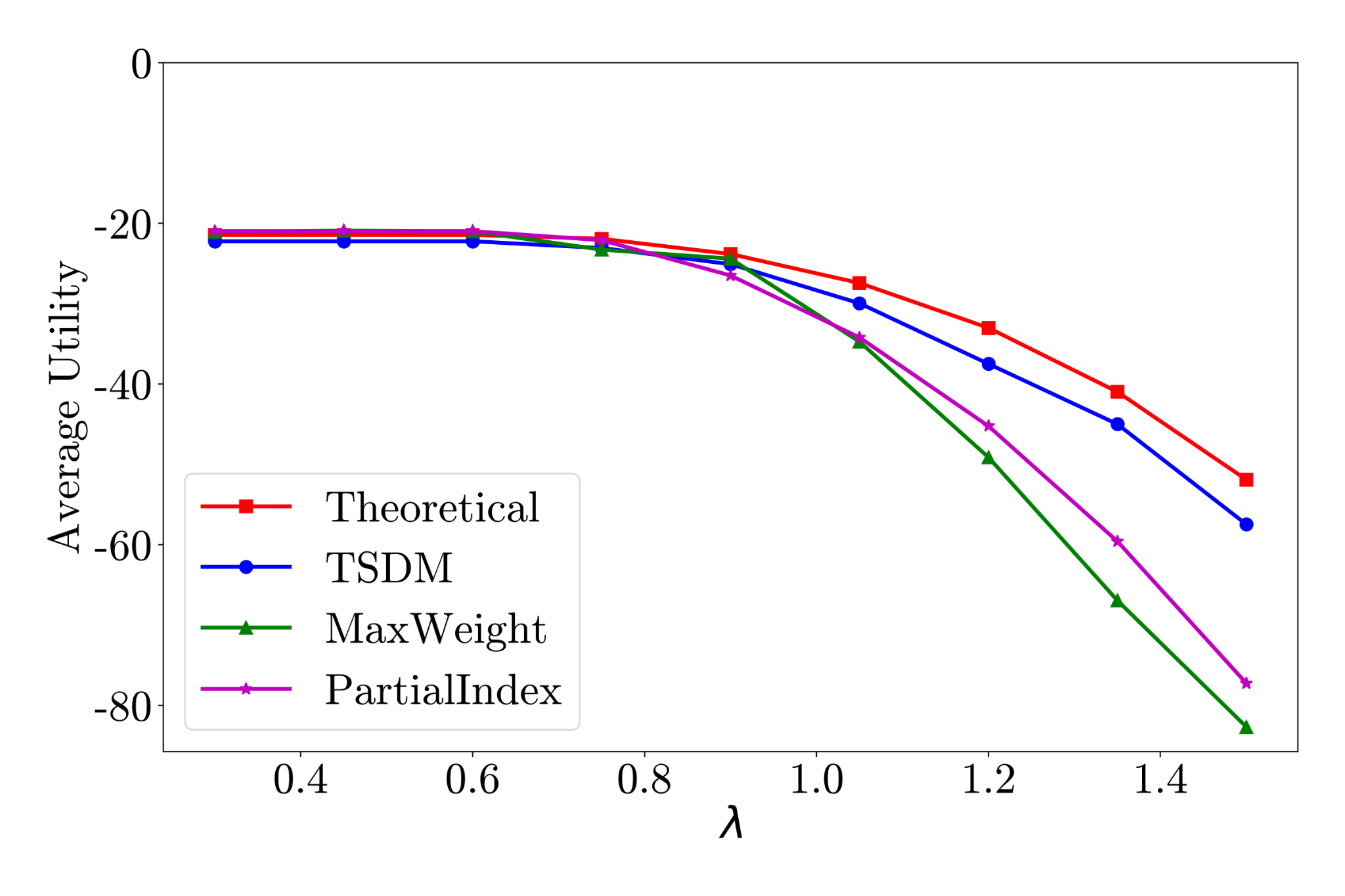}
       \label{fig:hoN50_2}
    }
    \caption{Average Utility when $\lambda$ changes in node-homogeneous systems.}
    \label{fig:stp}
\end{figure*}

We also consider a heterogeneous setting with $M=4$ channels and network sizes $N \in \{8,12,24\}$. The channel success probability matrix follows a cyclic right-shift pattern. Let the base probability vector be
$\mathbf{p}_1 = (0.90,0.70,0.50,0.20)$. Under this construction, each node has a distinct preferred channel, and $N/M$ nodes compete for each channel. The soft throughput constraint remains the same as in the previous settings. As $N$ increases from $8$ to $24$, the number of nodes competing per channel grows from $2$ to $6$, allowing us to examine how scheduling performance scales with increasing network contention under heterogeneous channel conditions. The simulation results are shown in Fig.~\ref{fig:heterogeneousshift}.
\begin{figure*}[t]
    \centering
    \subfigure[$N = 8,M = 4$.]{
       \includegraphics[width=0.3\linewidth]{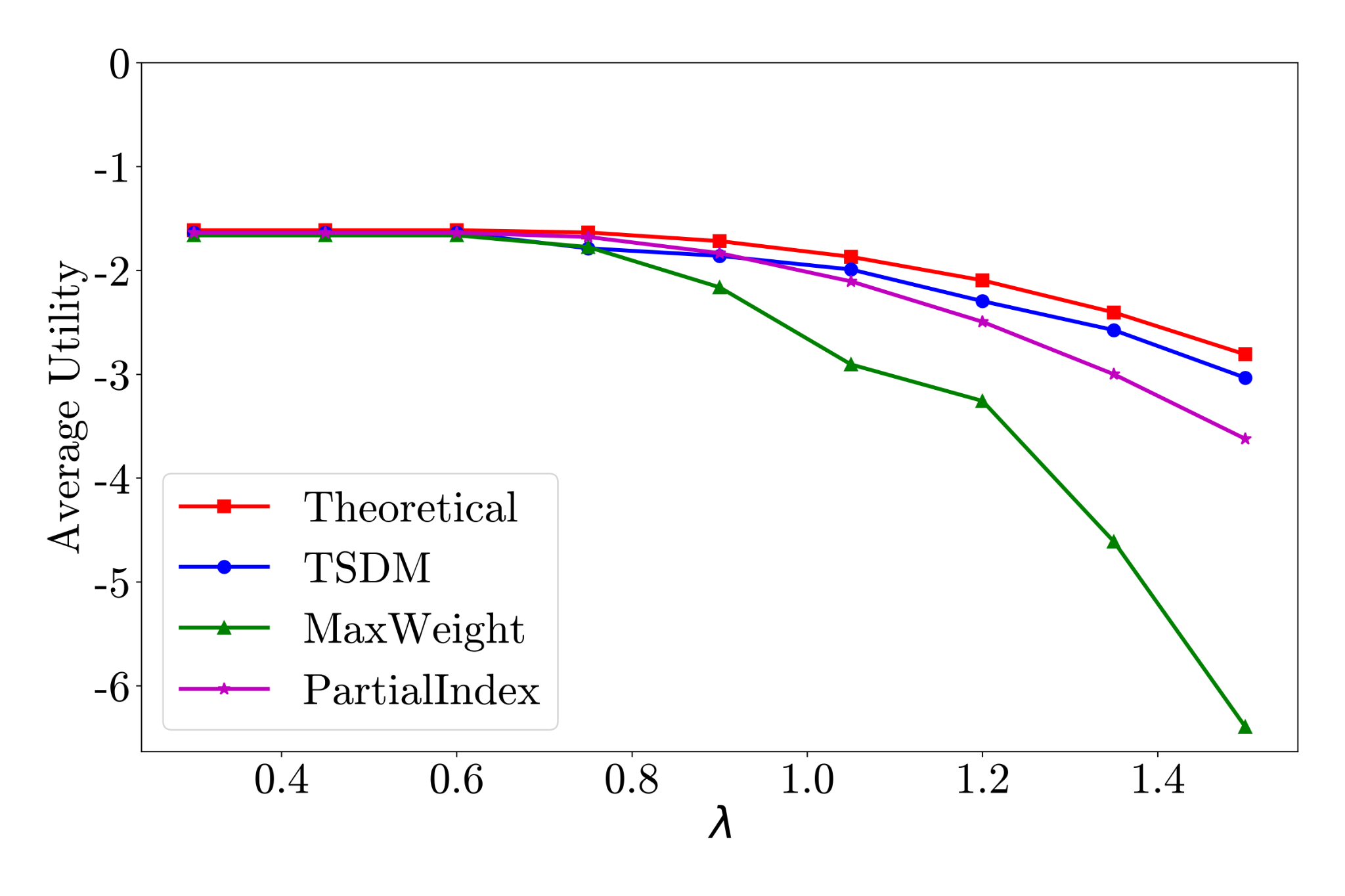}
       \label{fig:htN8_4}
    }
    \hfill
    \subfigure[ $N = 12,M = 4$.]{
       \includegraphics[width=0.3\linewidth]{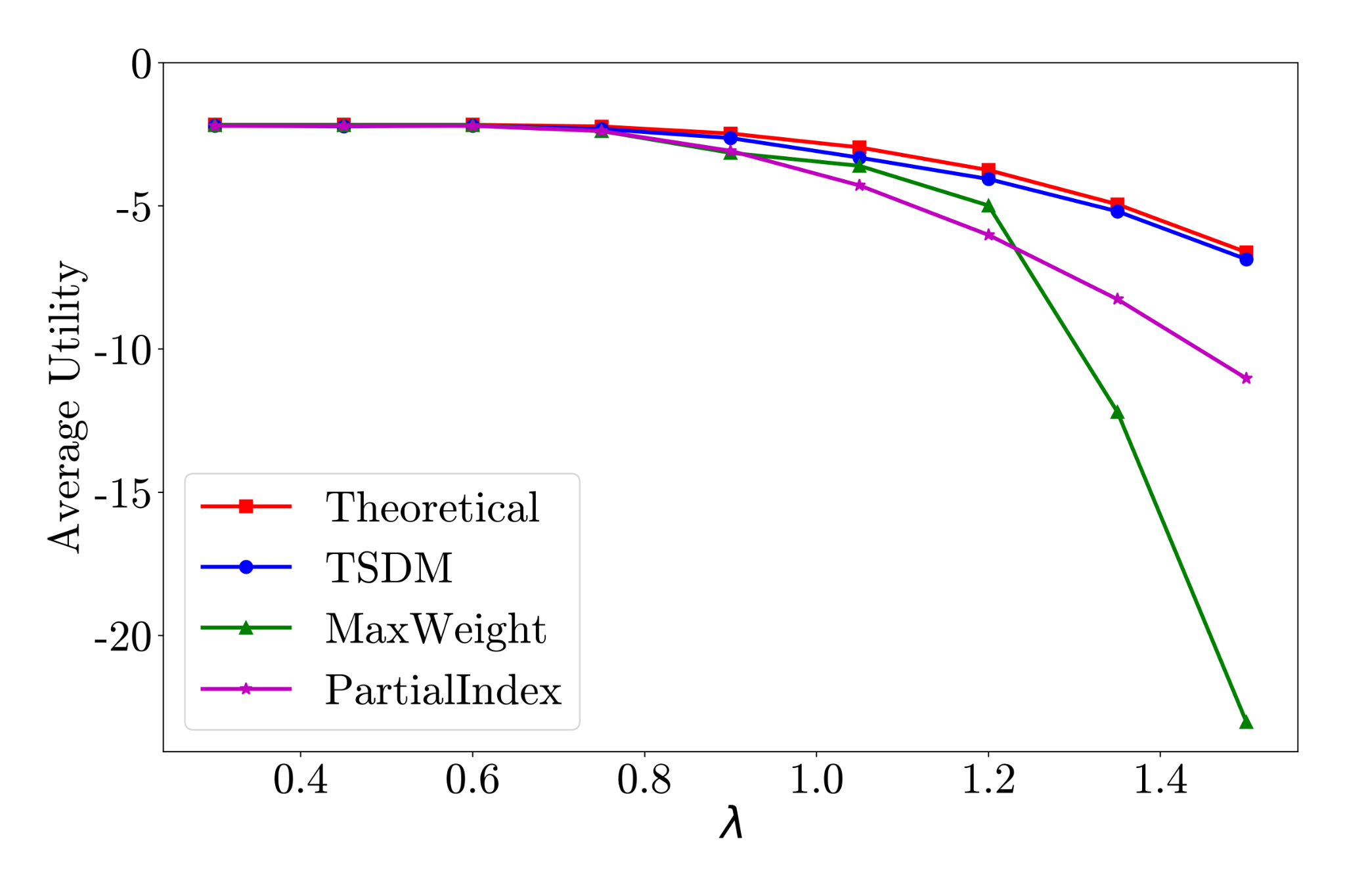}
       \label{fig:htN12_4}
    }
    \hfill
    \subfigure[$N = 24,M = 4$.]{
       \includegraphics[width=0.3\linewidth]{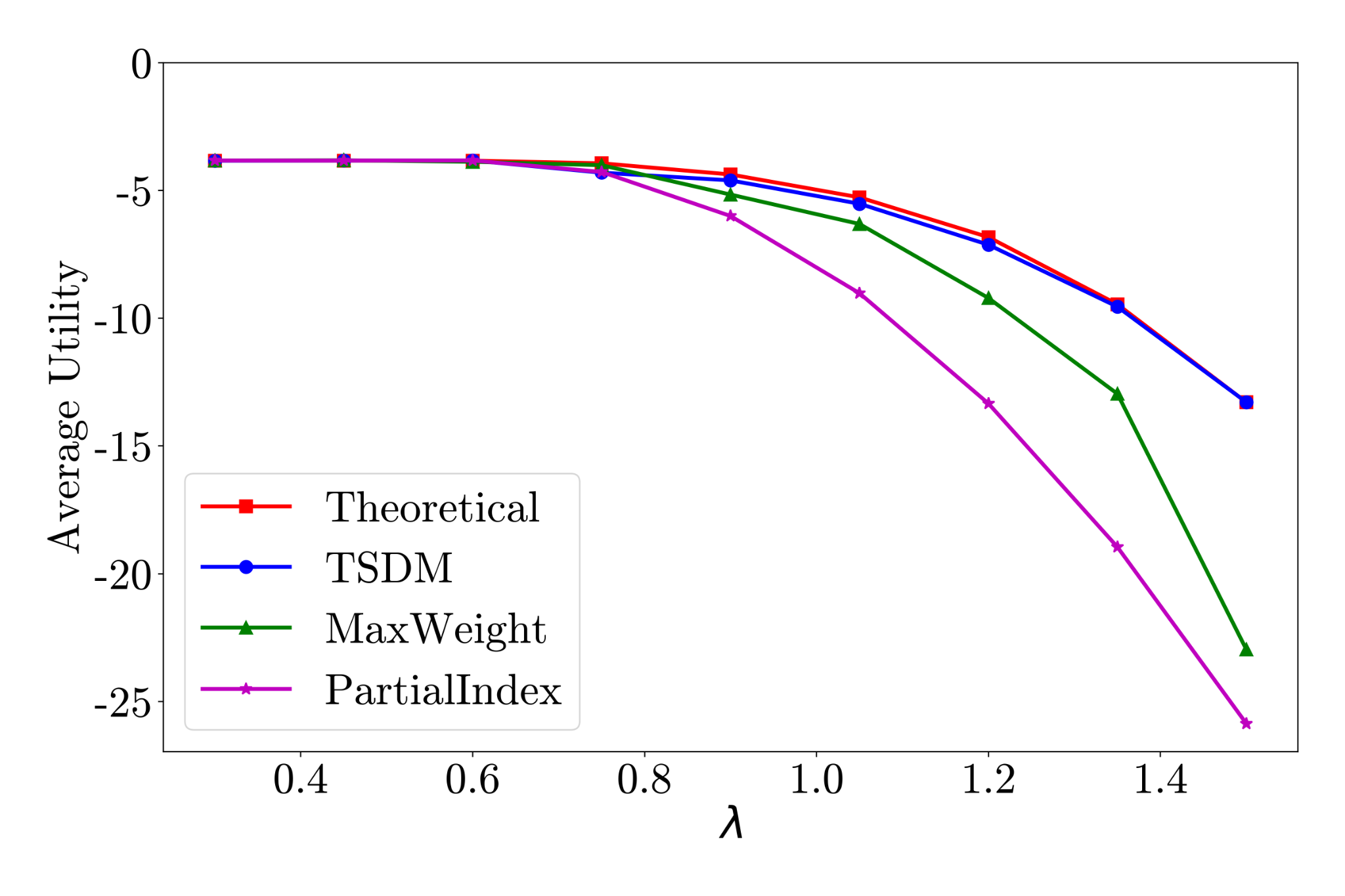}
       \label{fig:N24_4}
    }
    \caption{Average Utility when $\lambda$ changes in node-heterogeneous systems.}
    \label{fig:heterogeneousshift}
\end{figure*}

Across all configurations, TSDM consistently achieves the best performance and closely matches the theoretical upper bound. Moreover, TSDM significantly outperforms both the Max-Weight and Partial Index policies under high load. As $\lambda$ increases, satisfying the throughput requirements becomes increasingly difficult or even infeasible within the given channel capacity. The Max-Weight policy strictly penalizes throughput violations, forcing it to sacrifice AoI to meet demand, leading to poor AoI performance under high load. The Partial Index policy, on the other hand, ignores throughput requirements entirely, resulting in large throughput violations when $\lambda$ is large. In contrast, TSDM jointly optimizes both objectives and allows a controlled amount of throughput violation to achieve substantially lower AoI. 
\subsection{Weighted Proportional Fairness in both throughput and AoI}

We now evaluate the weighted proportional fairness problem described in Example~\ref{proportional_fairness}. In addition to evaluating the empirical performance of the proposed TSDM policy and its theoretical value, we also extend a classical proportional-fair(PF) scheduling policy as a baseline.
The classic PF policy schedules the node that maximizes $\frac{p_i(t)}{m_i(t)}$ in single-channel systems.
It is well known that this rule asymptotically maximizes $\sum_i \log m_i$, thereby achieving proportional fairness in throughput\cite{kushner2004convergence}.
To incorporate AoI into the baseline for comparison, we further extend it to the PF-MaxWeight policy by assigning matching weights of the form $p_{ij}(t) \left( \frac{1}{m_i(t)} - \frac{1}{a_i(t)} \right).$ This modified metric heuristically balances throughput fairness and AoI fairness and serves as a baseline for comparison. 

We first evaluate the weighted utility objective under homogeneous node settings with $M \in \{2,3,5\}$ channels. In all settings, nodes are divided into two equal groups with heterogeneous weights: the first $\lfloor N/2 \rfloor$ nodes are assigned $\alpha_i = 20, \beta_i = 1$ (throughput-oriented), while the remaining nodes are assigned $\alpha_i = 1, \beta_i = 20$ (freshness-oriented). The channel success probabilities are identical across all nodes. For $M=2$, we consider $N \in \{4,6,8,10,12,16\}$ with $p_{ij}=(0.9,0.3)$ for all $i$. For $M=3$, we consider $N \in \{6,9,12,15,18,24\}$ with $p_{ij}=(0.9,0.5,0.1)$ for all $i$. For $M=5$, we consider $N \in \{10,15,20,25,30,40\}$ with $p_{ij}=(0.9,0.7,0.5,0.3,0.1)$ for all $i$.  The simulation results are shown in Fig.~\ref{fig:weightedutility}.
\begin{figure*}[!t]
    \centering
    \subfigure[$M = 2$.]{
       \includegraphics[width=0.3\linewidth]{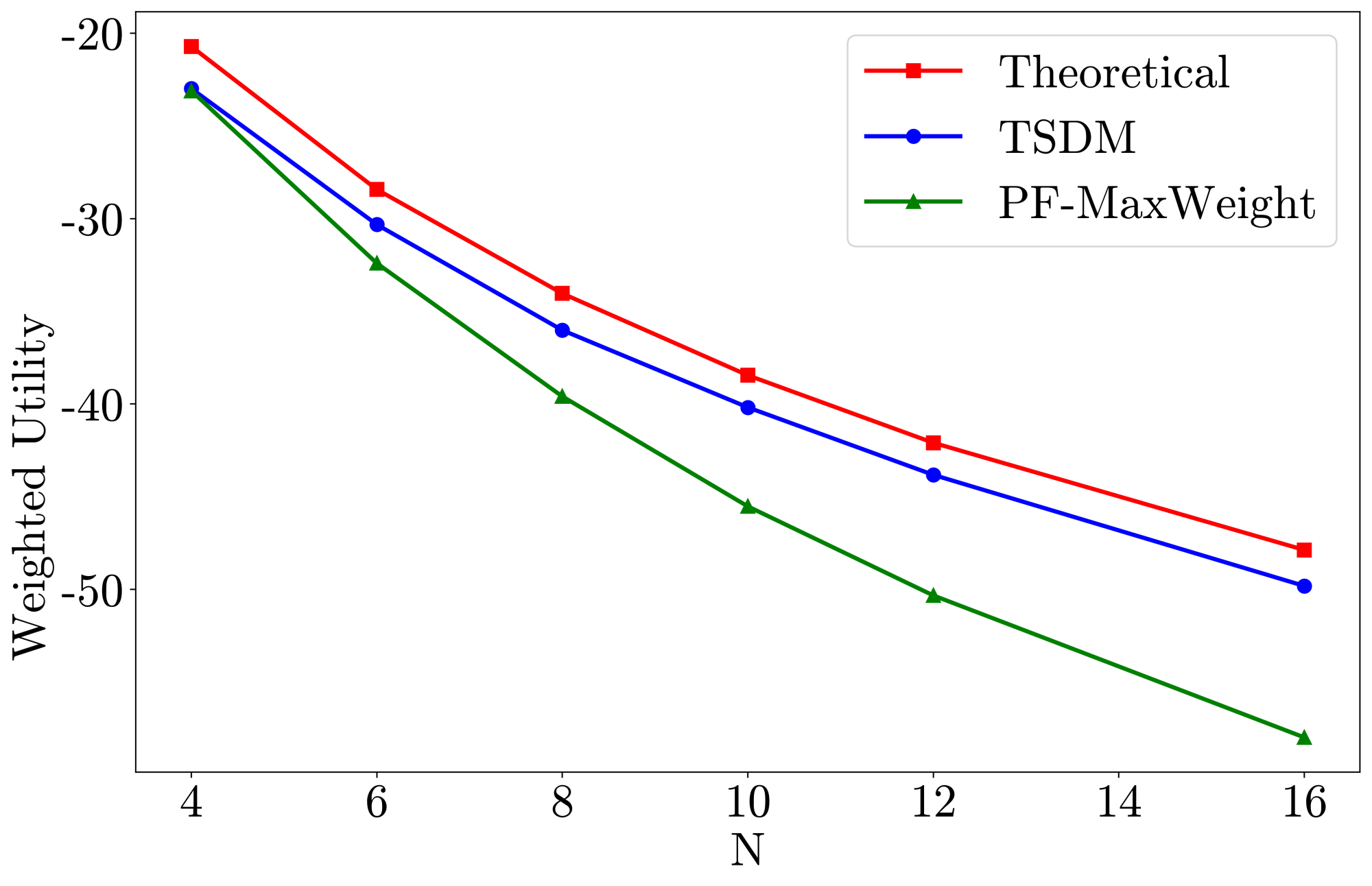}
       \label{fig:m2}
    }
    \hfill
    \subfigure[ $M = 3$.]{
       \includegraphics[width=0.3\linewidth]{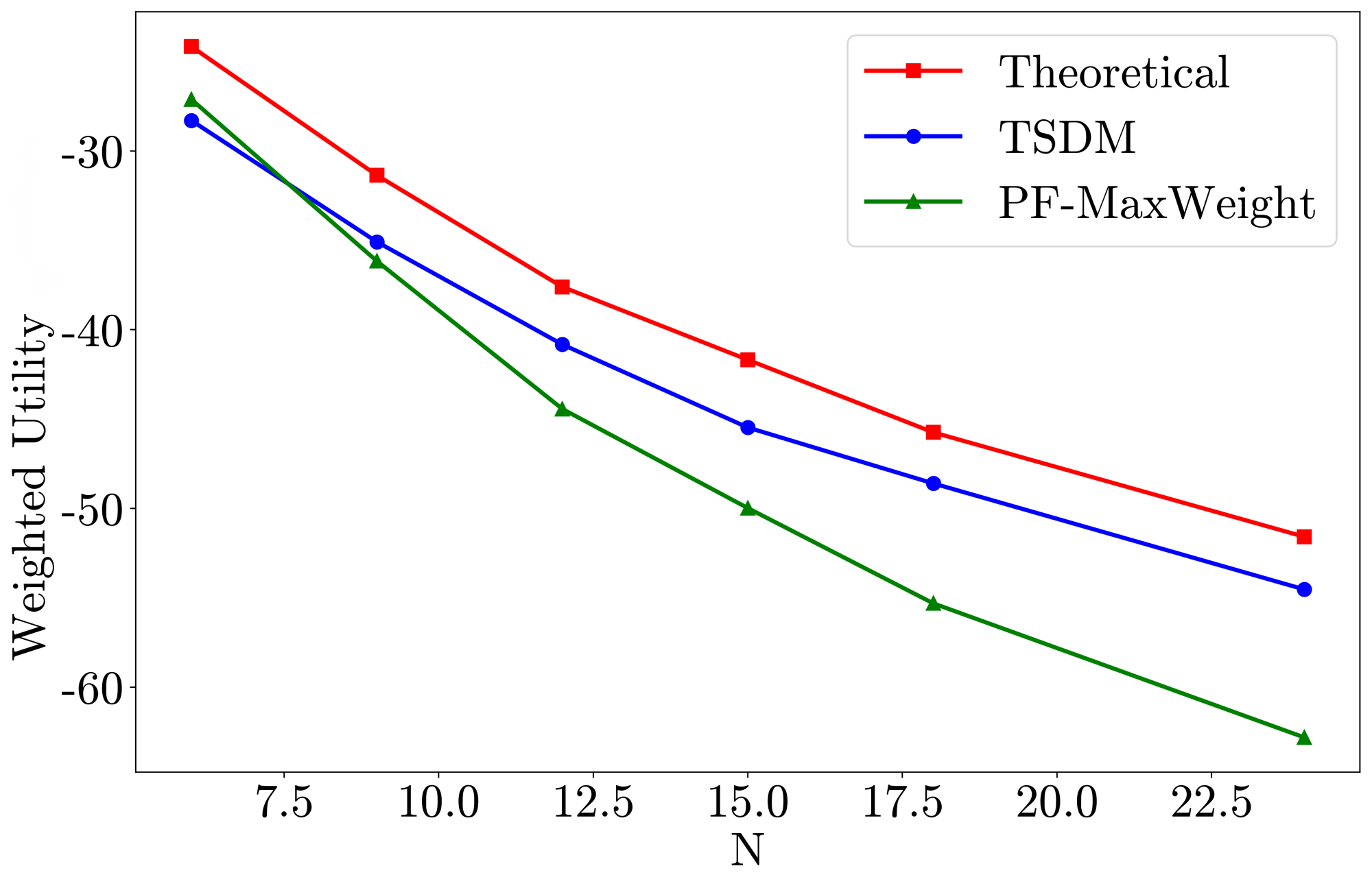}
       \label{fig:m3}
    }
    \hfill
    \subfigure[$M = 5$.]{
       \includegraphics[width=0.3\linewidth]{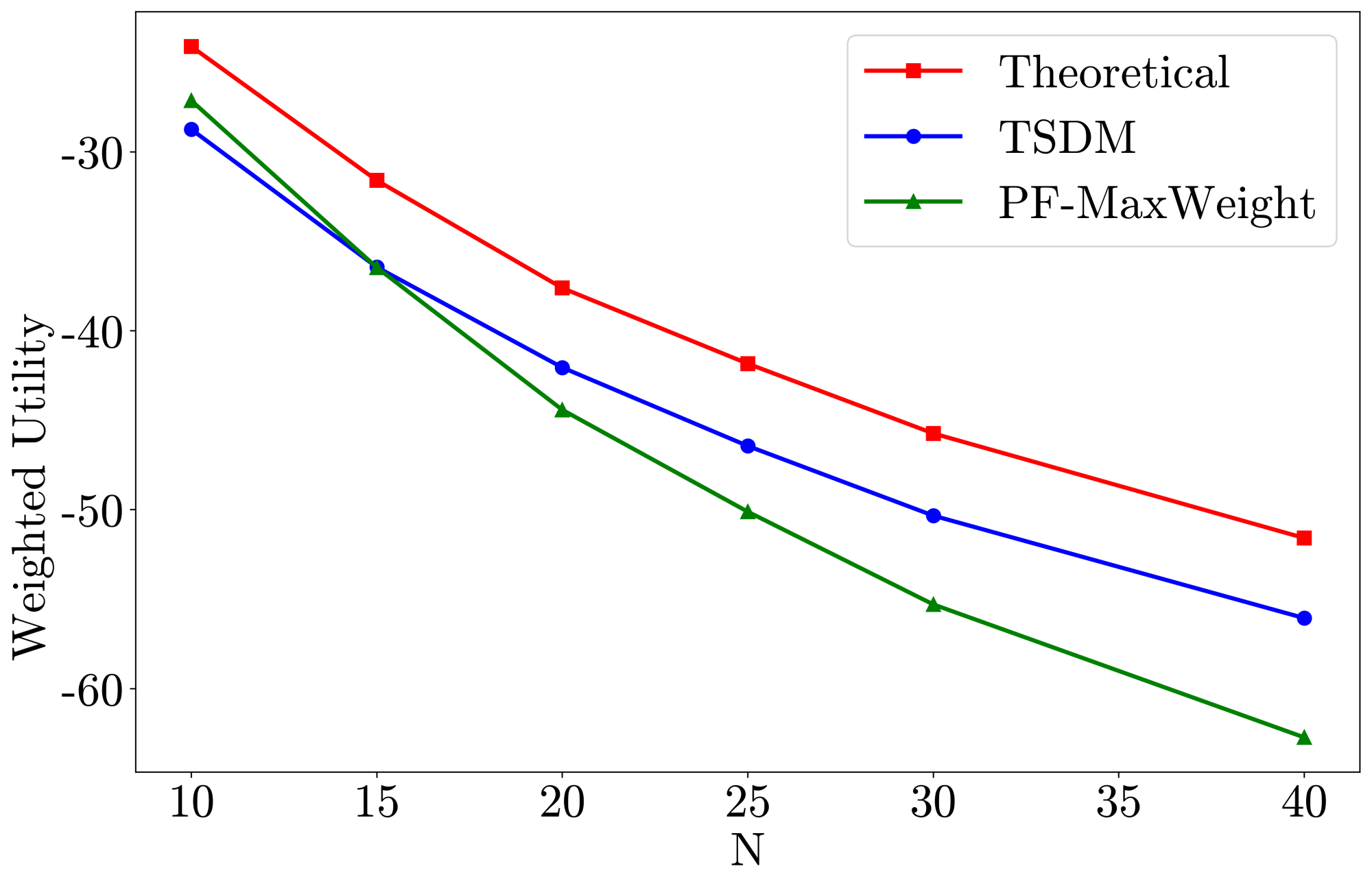}
       \label{fig:m5}
    }
    \caption{Weighted Utility when $N$ changes in node-homogeneous systems.}
    \label{fig:weightedutility}
\end{figure*}

We then evaluate the weighted utility objective under heterogeneous channel conditions using a cyclic shift pattern (as in Section~\ref{subsection:softp}) with $M\in\{3,4,5\}$ channels. The weights $(\alpha_i, \beta_i)$ are assigned based on the node's primary channel quality: nodes whose primary channel is strong are assigned a higher AoI weight ($\beta_i > \alpha_i$), while nodes with a weaker primary channel are assigned a higher throughput weight ($\alpha_i > \beta_i$), with some nodes assigned equal weights ($\alpha_i = \beta_i$). The simulation parameters for each case are as follows. For $M=3$: $N\in\{6,9,\ldots,24\}$ and $\mathbf{p}_1=(0.9,0.5,0.1)$. For $M=4$: $N\in\{8,12,\ldots,32\}$ and $\mathbf{p}_1=(0.9,0.7,0.3,0.1)$. For $M=5$: $N\in\{10,15,\ldots,40\}$ and $\mathbf{p}_1=(0.9,0.7,0.5,0.3,0.1)$. Results are shown in Fig.~\ref{fig:htweightedutility}.
\begin{figure*}[t]
    \centering
    \subfigure[$M = 3$.]{
       \includegraphics[width=0.3\linewidth]{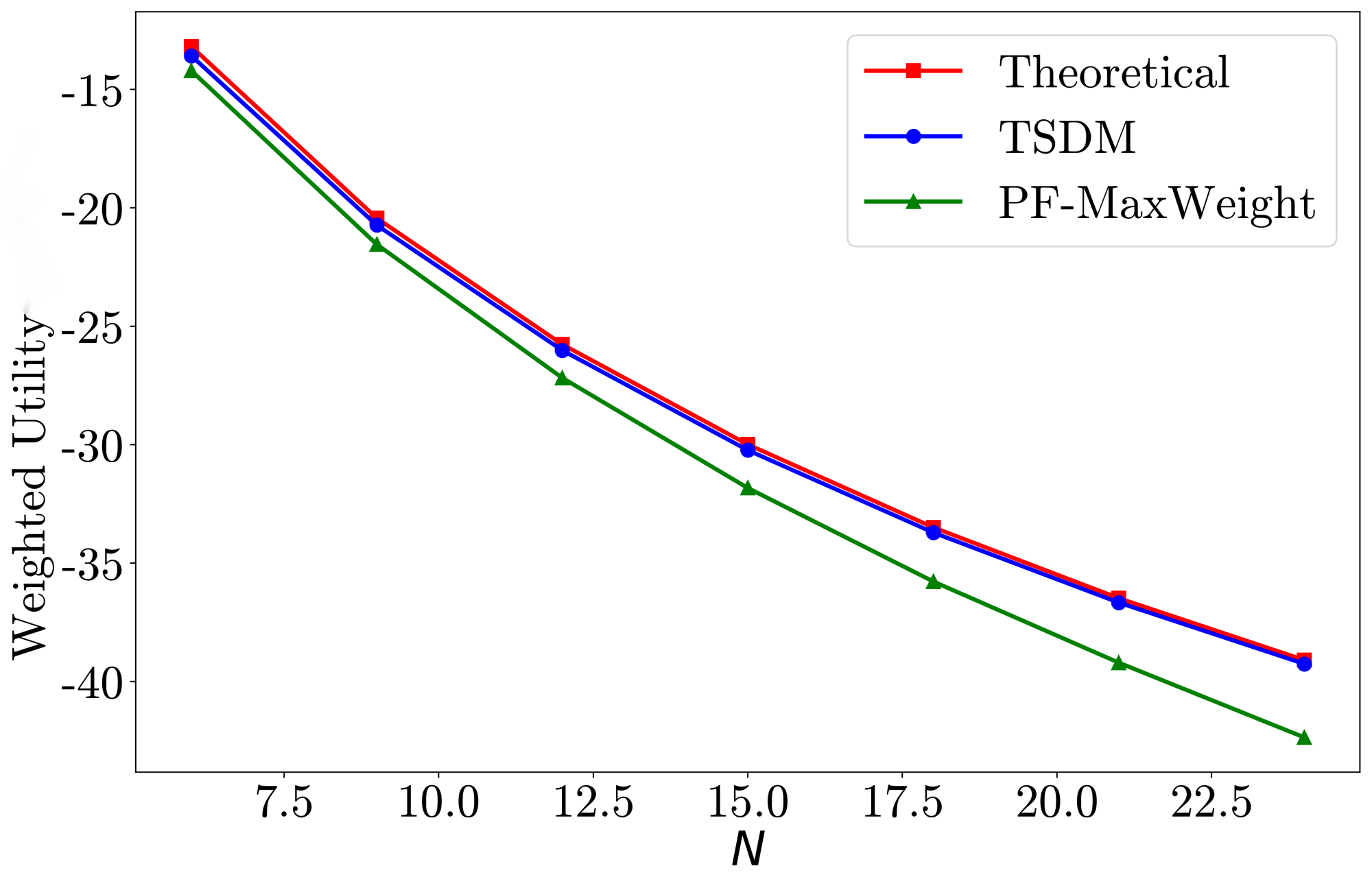}
       \label{fig:m3_ht}
    }
    \hfill
    \subfigure[ $M = 4$.]{
       \includegraphics[width=0.3\linewidth]{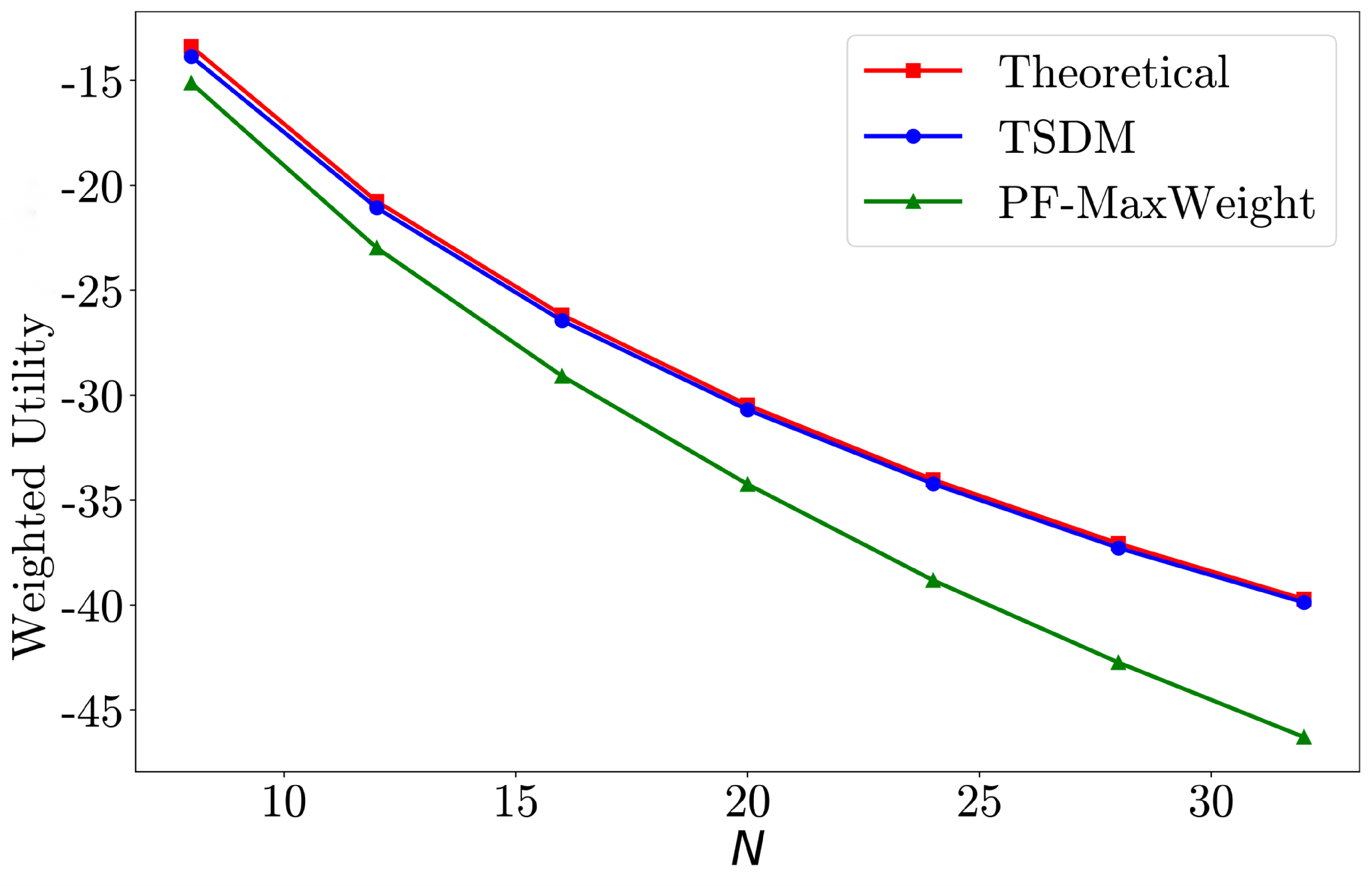}
       \label{fig:m4_ht}
    }
    \hfill
    \subfigure[$M =5$.]{
       \includegraphics[width=0.3\linewidth]{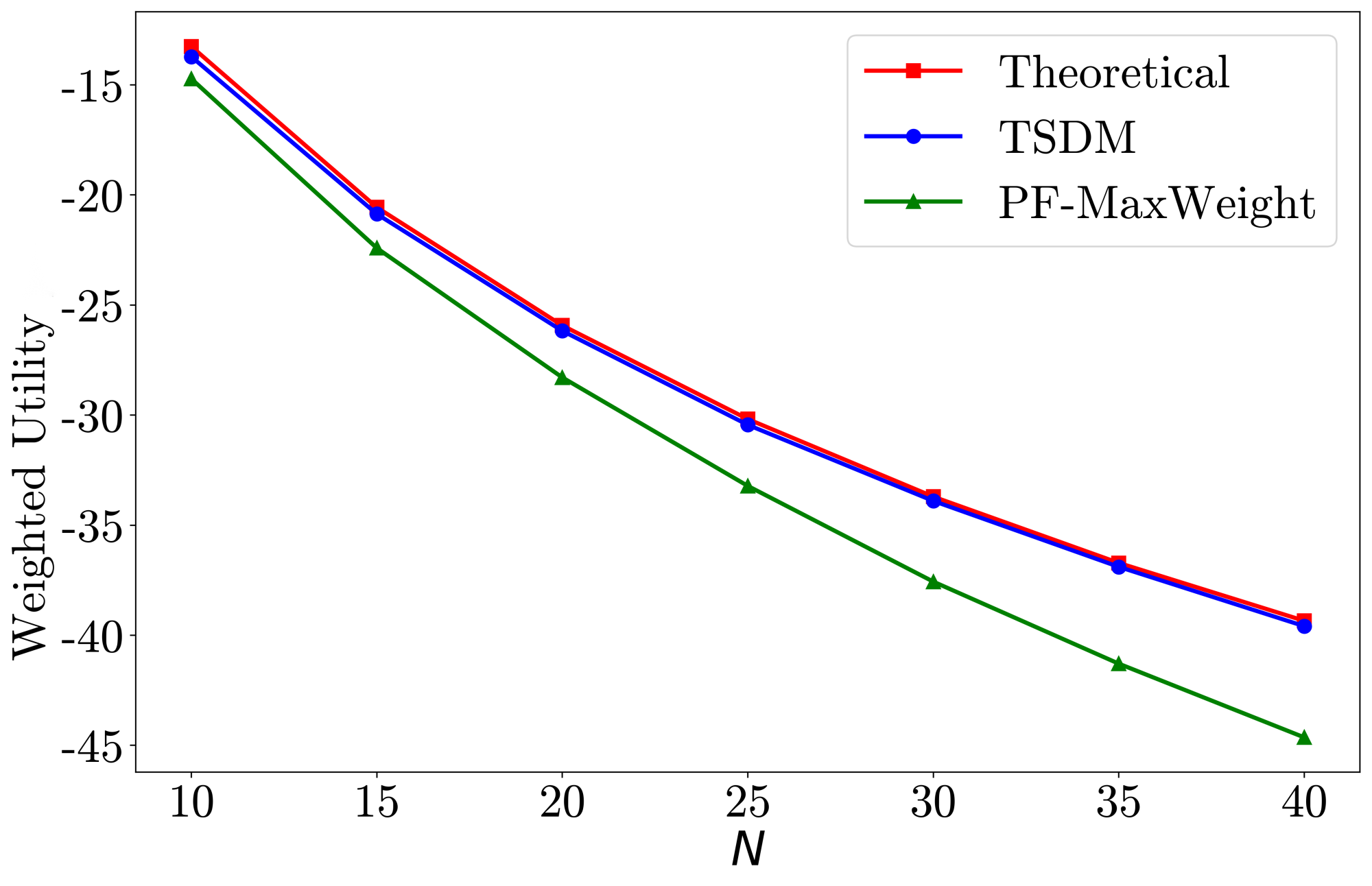}
       \label{fig:m5_ht}
    }
    \caption{Weighted Utility when $N$ changes in node-heterogeneous systems.}
    \label{fig:htweightedutility}
\end{figure*}

For both homogeneous and heterogeneous settings, TSDM achieves performance very close to the theoretical value and consistently outperforms the PF-MaxWeight baseline. 
\section{CONCLUSIONS}\label{section: conclusion}

In this paper, we studied the joint throughput-AoI optimization problem in multichannel wireless networks with heterogeneous and unreliable channels. We proposed a two-stage scheduling framework, called TSDM, to jointly optimize data freshness and throughput under practical wireless constraints. The proposed framework first maps desired throughput-AoI requirements to target transmission statistics using a second-order analytical model. Based on this mapping, the TSDM policy dynamically assigns channels to nodes through the WMD rule. We provided theoretical guarantees showing that TSDM ensures system stability and asymptotically achieves the desired throughput and AoI targets. Extensive simulations under both homogeneous and heterogeneous channel conditions demonstrated that TSDM achieves performance close to theoretical benchmarks and consistently outperforms existing scheduling baselines. Future work includes extending the proposed framework to more general network models, such as systems with time-correlated channels, unknown traffic arrivals, or distributed scheduling architectures.

\section{Acknowledgement}
This material is based upon work supported in part by NSF under Award Number CCF-2332800, the U.S. Army Contracting Command-Aberdeen Proving Ground under Grant Number W911NF-22-1-015, U.S. Army Combat Capabilities Development Command under Grant Number W911NF2520046, and the U.S. Office of Naval Research under Grant Number N000142412615.
\bibliographystyle{IEEEtran}
\bibliography{reference}
\end{document}